\documentclass[letterpaper,twocolumn,aps,prl,floatfix,superscriptaddress,footinbib,preprintnumbers]{revtex4-1}

\usepackage[utf8]{inputenc}
\usepackage{amsmath}
\usepackage{amssymb}
\usepackage{amsfonts}
\usepackage{amsthm}
\usepackage{physics}
\usepackage{qcircuit}
\usepackage{xcolor}
\usepackage[colorlinks=true,citecolor=blue,linkcolor=purple]{hyperref}
\usepackage{mathtools}
\usepackage[normalem]{ulem}
\usepackage{algorithm}
\usepackage{algpseudocodex}

\usepackage{xr}
\DeclarePairedDelimiterXPP\bigo[1]{\mathcal{O}}{(}{)}{}{#1}
\DeclarePairedDelimiterXPP\bigomega[1]{\Omega}{(}{)}{}{#1}
\DeclarePairedDelimiterXPP\bigtheta[1]{\Theta}{(}{)}{}{#1}

\usepackage[capitalise,compress]{cleveref}
\crefname{section}{Sec.}{Secs.}
\crefrangelabelformat{equation}{\textup{(#3#1#4)}--\textup{(#5#2#6)}}

\DeclareMathOperator{\poly}{poly}

\graphicspath{{./figures/}}

\newcommand{\QuICS}{
Joint Center for Quantum Information and Computer Science, NIST/University of Maryland, College Park, Maryland 20742, USA}
\newcommand{\JQI}{
Joint Quantum Institute, NIST/University of Maryland, College Park, Maryland 20742, USA}
\newcommand{\UMDCS}{Department of Computer Science and Institute for Advanced Computer Studies,
University of Maryland, College Park, Maryland 20742, USA}
\newcommand{\Harvard}{Department of Physics, Harvard University, Cambridge, Massachusetts 02138, USA}

\begin{document}
\title{Optimal Routing Protocols for Reconfigurable Atom Arrays}
\date{\today}
\author{Nathan~Constantinides}
\email{nconstan@umd.edu}
\affiliation{\JQI}
\affiliation{\QuICS}
\author{Ali~Fahimniya}
\affiliation{\JQI}
\affiliation{\QuICS}
\author{Dhruv~Devulapalli}
\affiliation{\JQI}
\affiliation{\QuICS}
\author{Dolev~Bluvstein}
\affiliation{\Harvard}
\author{Michael~J.~Gullans}
\affiliation{\QuICS}
\affiliation{\UMDCS}
\author{J.~V.~Porto}
\affiliation{\JQI}
\author{Andrew~M.~Childs}
\affiliation{\QuICS}
\affiliation{\UMDCS}
\author{Alexey~V.~Gorshkov}
\affiliation{\JQI}
\affiliation{\QuICS}

\begin{abstract}
    Neutral atom arrays have emerged as a promising platform for both analog and digital quantum processing. Recently, devices capable of reconfiguring arrays during quantum processes have enabled new applications for these systems. Atom reconfiguration, or \emph{routing}, is the core mechanism for programming circuits; optimizing this routing can increase processing speeds, reduce decoherence, and enable efficient implementations of highly non-local connections. In this work, we investigate routing models applicable to state-of-the-art neutral atom systems. With routing steps that can operate on multiple atoms in parallel, we prove that current designs require %
    $\bigomega{\sqrt N \log N}$ steps 
    to perform certain permutations on 2D arrays with $N$ atoms and provide a protocol that achieves routing in %
    $\bigo{\sqrt N \log N}$ steps for any permutation. We also propose a simple experimental upgrade and show that it would reduce the routing cost to $\bigtheta{\log N}$
    steps. 
\end{abstract}
\maketitle

\textit{Introduction.---}Arrays of optically trapped neutral atoms have demonstrated impressive capabilities as a platform for quantum simulation and quantum information processing~\cite{bluvstein_controlling_2021,cong_enhancing_2024,semeghini_probing_2021,ebadi_quantum_2021,ebadi_quantum_2022,hashizume_deterministic_2021,graham_multiqubit_2022,radnaev_universal_2024,browaeys_manybody_2020}. 
In contrast to architectures with fixed connectivity~\cite{bruzewicz_trappedion_2019,kjaergaard_superconducting_2020}, atom arrays can be dynamically reconfigured mid-computation, enabling parallelization and fast implementation of 
computations and primitives for fault tolerance
~\cite{beugnon_twodimensional_2007,bluvstein_quantum_2022,bluvstein_logical_2024,xu_constantoverhead_2024}. %
Crucial to the success of these applications is their fast and efficient circuit synthesis using native interactions allowed by the hardware.  Recent works have investigated compilers tailored to the unique capabilities of neutral atom arrays~\cite{wang_qpilot_2024,tan_compiling_2024,patel_geyser_2022a,patel_graphine_2023,ludmir_parallax_2024,tan_compilation_2024_enola}. %
In current designs, two-qubit gates are implemented using a globally illuminating laser that simultaneously couples pairs of atoms through Rydberg excitations if they are transported within the Rydberg interaction radius $R_b$. Compilers make use of three steps: implementing gates using the Rydberg interactions, mapping qubits to atoms, and  routing qubits (i.e., reconfiguring atoms). 
Fast implementation of the atom reconfiguration step enables a more efficient implementation of parallel two-qubit gates between arbitrary pairs of atoms. %

\begin{figure}[b]
    \includegraphics[width=1\linewidth]{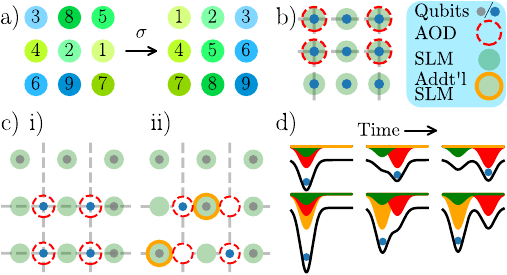}
    \caption{\textbf{(a)} An example of an instance of a routing problem on $N=9$ qubits in a 2D lattice. \textbf{(b)} Schematic of atoms (smaller blue dots) in a combination of SLM-generated traps (light green disks) and AOD-generated traps (dashed red circles). %
    \textbf{(c)} Schematic of two scenarios for selecting atoms for transfer: \textbf{(i)} Grid transfer, where atoms trapped in a static SLM array are selected and transferred by an a AOD-defined grid. \textbf{(ii)} Selective transfer, where atoms trapped in a static SLM array are (de)selected  by a second SLM and thus excluded from being transferred to an AOD-defined array. Locations of the deselected sites are marked by the orange circles. \textbf{(d)} Time-series of SLM, AOD, and additional SLM trap potentials (green, red, and orange curves, respectively), as well as the total trap potential (black line) for selected (upper row) and deselected (lower row) sites during a selective transfer. The atom, following the potential minimum, moves to AOD trap on selected sites but stays in SLM trap otherwise.} %
    \label{fig:main}
\end{figure}

Atom reconfiguration is a generalization of quantum routing, or performing a permutation $\sigma\colon S \to S$ on the set of $N$ qubits labeled with $S = \{0, \dots, N-1\}$. Quantum routing has been well-studied in architectures with connectivity constraints described by a static coupling graph~\cite{cowtan_qubit_2019,childs_circuit_2019,bapat_quantum_2021,yuan_full_2024}. However, the coupling graph model does not represent the connectivity of reconfigurable atom arrays as a single time-step can involve gates or swaps between qubits that were previously far away and have been transported close to each other during the step. Therefore, there is a need for new routing models that apply to the specific capabilities of these platforms.

In this paper, we investigate routing models inspired by experiments with one- and two-dimensional arrays of reconfigurable neutral atoms. Figure~\ref{fig:main}(a) depicts routing on a 2D rectangular array that performs a permutation $\sigma$ and rearranges the atoms (circles) into a desired order. All of our models are motivated by experiments in which the atoms are trapped in optical tweezers generated by a stationary Spatial Light Modulator (SLM) and, during each step, are rearranged using movable Acousto-Optic Deflector (AOD) optical tweezers. In Fig.~\ref{fig:main}(b), the green circles represent the SLM traps that hold the atoms (blue dots). AOD traps (red dashed circles) %
form on a grid of intersecting AOD rows and columns (gray dashed lines). %
The positions of the rows and columns can be continuously adjusted, with the constraint that they do not cross. Atoms trapped at the intersections can be coherently transported within this dynamically adjustable grid that maintains site order.

Two possibilities for the transfer of a subset of atoms in an SLM array to and from a rectangular grid of AOD traps are shown in Fig.~\ref{fig:main}(c). In a \emph{grid transfer} [Fig.~\ref{fig:main}(c)(i)], all the atoms in the SLM traps next to the AOD grid are transferred to the AOD. This is how AODs are utilized in current designs. In contrast, we propose \emph{selective transfers} [Fig.~\ref{fig:main}(c)(ii)], with which one can choose to transfer only a subset of atoms from such SLM traps to a subset of the AOD grid. 

We prove lower bounds for routing on 2D arrays using either grid or selective transfers to load atoms between the SLM and AOD traps. We define each single step of routing to be a rearrangement of atoms in the SLM grid with a constant number of transfers between the SLM and a movable AOD. Each such step can involve parallel, simultaneous operations on multiple atoms. With this framework, we show that routing with grid transfers on an array with $N$ atoms requires $\bigomega{\sqrt{N} \log{N}}$ steps for most permutations. For routing with selective transfers, however, the corresponding lower bound is only $\bigomega{\log N}$. Additionally, we provide routing protocols that saturate these bounds up to constant factors, achieving routing with grid transfers and selective transfers in $\bigo{\sqrt{N}\log N}$ and $\bigo{\log N}$ steps, respectively. Therefore, selective transfers significantly speed up routing on 2D arrays. For the case of routing with grid transfers, we provide a protocol that performs any sparse permutation---one with at most $\poly(\log N)$ non-identity elements per row (or column)---in $\poly(\log N)$ steps, improving the $\bigo{\sqrt N \log N}$ result, which holds for a generic permutation.
\textit{Experimental details.---}Neutral atom platforms typically employ two types of optical tweezers: SLMs that have the capability to generate arbitrary 2D arrays of traps~\cite{bergamini_holographic_2004,barredo_atomatom_2016,barredo_synthetic_2018} by modulating the amplitude or phase of an incident beam, but do not currently have the ability to transport atoms during the execution of a quantum computation due to their slow update rate; and AODs that are formed at the intersection of two sets of rows and columns and therefore, are constrained to grids of traps \cite{barredo_atomatom_2016, bluvstein_quantum_2022}. %
These rows and columns may be steered independently, thus deforming their rectangular grid, but no two columns or rows may cross to prevent atom collisions or interference between the radio frequency tones driving the device. %

Both approaches we consider for selecting and transferring atoms from a stationary SLM to a movable AOD are experimentally feasible.
For the grid transfer, which is shown in  Fig.\ \ref{fig:main}(c)(i) and is the current state-of-the-art~\cite{bluvstein_logical_2024}, %
the AOD trap potentials are turned on at the location of the selected SLM sites, and are made deeper than the SLM traps. 
The AOD traps are then moved away, and the selected atoms, following the lowest potential, transfer from SLM to AOD. %
For the selective transfer, shown in Fig.\ \ref{fig:main}(c)(ii) and (d), we consider a similar approach to grid transfer, with the use of an additional SLM array to select and ``deactivate" arbitrary sites from transferring into the movable AOD array. 
The additional SLM trap potentials are turned on at the grid sites that are to be excluded from the transfer, and are made deeper than both the primary SLM and AOD traps.  This ``deselects'' atoms in these SLM traps from moving with the AOD traps when the AOD traps %
are moved away while the rest of the atoms of the grid are transferred to the AOD traps.

For each routing model, we assume atoms are placed in an SLM array. We then define single routing steps that can be done with an AOD array. Every single step consists of a constant number of pickups and drop-offs as well as constant potential ramp-ups and ramp-downs of AOD rows and columns and takes the atoms back to the original SLM array. We then investigate the number of single steps needed to implement different permutations of the atoms and determine the maximum number of single steps for the implementation of any permutation. This maximum number is the model-specific \emph{routing number}, which captures the worst-case time to perform a permutation in that model [see Def.~\ref{supp-def:routingnumber} in the Supplemental Material (SM)~\footnote[1]{See Supplemental Material for more details about our models, formal definitions, detailed proofs of theorems, as well as algorithms for hypercube routing and sparse routing}].

\textit{1D routing.---}Some recently performed and proposed neutral atom experiments involve operations on atoms stored in a one-dimensional chain of static traps, or require routing of 1D chains of atoms as an experimental subroutine~\cite{xu_constantoverhead_2024, hashizume_deterministic_2021}. Therefore, before moving to 2D, we consider the problem of permuting $N$ atoms arranged in a 1D chain of static SLM traps. We define a single routing step to be an in-order swap of two equal-sized subarrays of the 1D chain (see Def.~\ref{supp-def:1dSwapStep} in SM~\footnotemark[1]). An example of a possible implementation of a single routing step in this model is shown in Fig.~\ref{fig:1DSwapDesc}. A single step is given by two disjoint ascending ordered sequences $A = \{a_i\}$ and $B = \{b_i\}$ (with $|A| = |B|$) and swaps each $a_i$ with $b_i$ while leaving the other atoms in place. This single step involves four transfers of atoms between traps, and is implemented purely via motion and AOD-SLM tweezer transfers. The two subarrays $A$ and $B$ (orange and blue atoms in Fig.~\ref{fig:1DSwapDesc}, respectively), are swapped by picking up $A$ and replacing each atom in $B$ pairwise in order [Fig.~\ref{fig:1DSwapDesc}(a-b)]. The vacant SLM traps left by $A$ are likewise replaced by those atoms from $B$ [Fig.~\ref{fig:1DSwapDesc}(b-c)]. $A$ and $B$ are constrained to remain in order to prevent crossing of AOD tones and atom collision. %

\begin{figure}
    \includegraphics[width=0.95\linewidth]{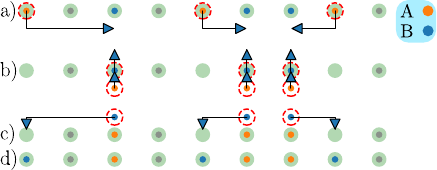}
    \caption{Step-by-step visualization of an in-order swap of two subsets of atoms of a 1D chain. Orange and blue dots are the two subsets involved in the swap, and gray dots are the rest of the atoms. Green circles are the SLM traps that hold the atoms before and after the swap. Red dashed circles are AOD traps that are used to transport and swap the atoms. %
    \textbf{(a)} First, the first subset of atoms involved in the swap, $A$, are picked up from SLM traps by the AOD traps and are transported next to the second subset of atoms, $B$. \textbf{(b)} Then, using an additional row of AOD traps, $A$ ($B$) atoms are moved into (out of) the SLM traps. \textbf{(c)} Next, $B$ atoms are transported back to the original SLM traps occupied by $A$ atoms. \textbf{(d)} The final arrangement of the atoms after the swap.
    }
    \label{fig:1DSwapDesc}
\end{figure}

\textcite{xu_constantoverhead_2024}\ also describe a model of routing on 1D chains of atoms in which one generates $N/2$ extra empty SLM traps in addition to those $N$ storing atoms. In this model, %
each step utilizes only a constant number of AOD-SLM transfers and AOD motion. The model's single step permutation is a riffle shuffle: an interleaving of two unique rising sequences, analogous to interleaving two parts of a deck of cards while shuffling them. %
With this single step, they show that a divide-and-conquer algorithm can implement any permutation in $\log_2 N$ steps. %

In analogy to the algorithm for routing with riffle shuffles, one can implement any permutation in $\log_2 N$ ordered swaps. Labeling the atoms by their destinations, one can transport the atoms $\{0, \dots, N/2-1\}$ to the left $N/2$ traps of the array in a single step by simply swapping the sets of qubits on the left and right halves of the array that are targeting the opposite side. Then, one can continue to route smaller sets of atoms in parallel in each partition until the permutation is completed (see Sec.~\ref{supp-sec:1d_protocols} in SM~\footnotemark[1]). %

The routing numbers for both riffle shuffle and ordered swap models of 1D routing are lower bounded by $\bigomega{\log N}$.
We show this (see Sec.~\ref{supp-sec:1d_bounds} in SM~\footnotemark[1]) by considering a general permutation $\sigma$ and a set $\mathcal R$ of $k$ atoms whose relative order is reversed by $\sigma$.
The proof follows from the fact that, in any riffle shuffle, all the atoms of the array are partitioned into two sets that retain their relative order. One of these partitions is guaranteed to include at least half of those atoms in $\mathcal R$. Therefore, after each routing step, at least $k/2$ atoms from $\mathcal R$ remain in reverse order. %
This implies that it takes at least $\log_2 k$ steps to route $\sigma$. For the reversal permutation $\sigma_r(i) = N - i - 1$, all atoms are reversed, so the worst case routing number is lower bounded by $\bigomega{\log N}$. When considering routing with ordered swaps, the proof is analogous.%

\begin{figure*}
    \centering
    \includegraphics[width=0.9\linewidth]{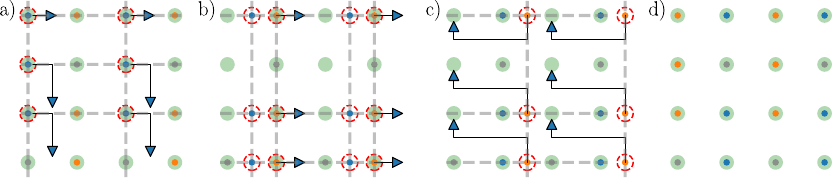}
    \caption{Step-by-step visualization of a swap of two combinatorial rectangles of a 2D array. See the caption in Fig.~\ref{fig:1DSwapDesc} for descriptions of symbols in the figure and of each panel.}
    \label{fig:2dswap}
\end{figure*}

\textit{2D routing.---}For 2D atom arrays, we identify natural single-step permutations for the cases of routing with grid and selective transfers. In the absence of selective transfers, we take any single routing step to be an in-order swap of two \emph{combinatorial rectangles} $R_1$ and $R_2$ of the same dimension (see Defs.~\ref{supp-def:rect},~\ref{supp-def:2DRectSwaps} in SM~\footnotemark[1]). A combinatorial rectangle $R$ %
is an array of points formed by the intersection of two sets of rows and columns $A$ and $B$, with $R = A \times B = \{(i,j) \mid i \in A,\, j \in B\}$. Its dimension is $(|A|, |B|)$. Each combinatorial rectangle is ordered lexicographically first by row and then by column, with the $i$th element of a rectangle denoted $[R]_i$. For example, $( (1, 1), (1, 2), (3, 1), (3, 2) )$ is a dimension $(2, 2)$ rectangle. A single routing step with grid transfers is the permutation $\sigma$ that swaps, for all $i$, the $i$th atoms of each rectangle, $[R_1]_i$ and $[R_2]_i$, and leaves all other atoms in place. The requirement that rectangles may be swapped only in order stems from the column and row non-crossing requirements of AODs. Figure~\ref{fig:2dswap} depicts the implementation of such a swap via atom motion.
It could also be implemented by a combination of Rydberg gates between the pairs of neighboring blue and orange atoms in Fig.~\ref{fig:2dswap}(b) and single-qubit gates before returning the blue atoms to their original SLM traps.

With selective transfers, a single routing step is a \emph{masked} in-order swap of two combinatorial rectangles (see Def.~\ref{supp-def:2DArbPickStep} in SM~\footnotemark[1]). That is, for two combinatorial rectangles $R_1$ and $R_2$, and some masking function %
$\mathcal M \colon \mathbb{Z} \to \{0, 1\}$, any two atoms $[R_1]_i$ and $[R_2]_i$ with $\mathcal M(i) = 1$ are swapped, and all other atoms are left in place. This single step could be implemented analogously to that seen in Fig.~\ref{fig:2dswap}(a-c), with the AOD-SLM transfers becoming selective transfers [see Fig.\ \ref{fig:main}(c)(ii)] to swap only the masked part of each rectangle of blue and orange atoms. %

Lower bounds on routing in both 2D models (i.e., with grid transfers and selective transfers) follow from a counting argument (see Sec.~\ref{supp-sec:2d_bounds} in SM~\footnotemark[1]). The number of single-step permutations for $N$ atoms on an $\sqrt N \times \sqrt N$ grid is $2^{O(N)}$ and $2^{O(\sqrt N)}$ with grid and selective transfers, respectively, while the number of permutations is $N! = 2^{\Theta(N \log N)}$. Since each permutation must have a unique routing schedule, the worst-case permutation requires $\Omega(\log N)$ steps with selective transfers, or $\Omega(\sqrt N \log N)$ steps with grid transfers.

We now present routing protocols for both 2D models matching their lower bounds up to constant %
factors (see Sec.~\ref{supp-sec:2d_protocols} in SM~\footnotemark[1]). These protocols borrow from the theory of routing on a coupling graph, and take advantage of a natural embedding of the $d$-dimensional hypercube graph $Q_d$ into the 2D grid \cite{bluvstein_logical_2024}. The vertices of $Q_d$ are $V = \{0,1\}^d$ (the set of all binary strings of length $d$), and its edges are $E = \{(i,j) \in V \times V \mid H(i, j) = 1\}$, where $H(i,j)$ is the binary Hamming distance, or number of bits where $i$ and $j$ differ.

The 2D embedding maps each of $N = 2^d$ atoms to a vertex of the hypercube $Q_d$ and positions it in 2D space based on its binary address. The atoms are arranged on a $\sqrt{2N} \times \sqrt{N / 2}$ grid for $d$ odd, or a $\sqrt{N} \times \sqrt{N}$ grid for $d$ even, shown in an example for $d=4$ in Fig.~\ref{fig:hcembedding}. Each vertex is mapped to the corresponding atom on the grid at coordinate $(i,j)$, with row $i$ and column $j$ the first $\lceil d / 2 \rceil$ and last $\lfloor d /2 \rfloor$ bits of its binary address, respectively. For example, the hypercube vertex ${\color{red} 10}{\color{blue}{01}}$ is mapped to coordinate $({\color{red}10}, {\color{blue}01})$ (or $(2, 1)$ in decimal).

With this embedding, one can choose a bit position $1 \leq k \leq d$, and with selective transfers swap in a single routing step any set of pairs of atoms $(i,j)$ whose binary addresses differ only on the $k$th bit. This is possible because the set $A$ of atoms whose $k$th bit is $0$, and set $B$ of atoms whose $k$th bit is $1$ each correspond to sets of atoms that are equal-dimension combinatorial rectangles in the hypercube embedding. Additionally, each $i \in A$, $j \in B$ that differ only on the $k$th bit also maintain the same lexicographic order in their respective rectangle, so they will be swapped with each other if an in-order swap of $A$ and $B$ is performed. By choosing an appropriate masking function $\mathcal M$, the single selective transfers routing step given by $A$, $B$, and $\mathcal M$ implements the swap of any set of pairs of atoms differing only on the $k$th bit. 

\begin{figure}[b]
    \centering
    \includegraphics[width=220pt]{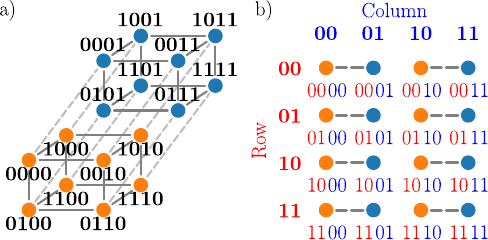}
    \caption{\textbf{(a)} The 4-dimensional hypercube graph $Q_4$. Vertices are given by orange and blue circles and are labeled by the binary address. For orange vertices the last bit is set to $0$ (e.g., $1110$), and for blue the last is set to $1$ (e.g., $0001$). Edges are given by solid or dashed gray lines. Dashed edges connect vertices that differ on the last bit. \textbf{(b)} The 2D embedding of $Q_d$ for $d = 4$. Atoms (orange and blue circles) are labeled by the corresponding graph vertex. The graph is embedded into a rectangular grid of $N=2^d$ atoms. %
    Each atom is assigned to a vertex of $Q_4$ by concatenating its binary row and column coordinates, shown by the red and blue labels. Similarly to (a), orange and blue atoms correspond to addresses with the last bit set to $0$ and $1$, respectively. Edges of the hypercube connecting atoms that differ on the last bit are depicted by gray dashed lines. Both the sets of orange and blue atoms form combinatorial rectangles. Therefore, with selective transfers, any set of pairs of orange and blue atoms connected by dashed lines may be swapped in a single step. In fact, any set of pairs of atoms whose addresses differ only on some $k$th bit can be swapped in a single step.} %
    \label{fig:hcembedding}
\end{figure}

This primitive operation is used to implement an optimal routing algorithm for selective transfers. Reference~\cite{alon_routing_1993} describes an algorithm for routing on the hypercube. 
We apply this algorithm to the problem of routing on neutral atoms with the hypercube embedding. The output of the algorithm on $Q_d$ can be scheduled so that it divides routing into a sequence of $2d - 1$ steps, where each step performs swaps across sets of edges $(i, j)$ of $Q_d$ where the addresses $i,j$ differ only on some $k$th bit. Since these are single routing steps with selective transfers, and $d=\log_2 N$, this realizes a routing procedure for selective transfers taking $2 \log_2 N - 1$ routing steps.

The output of this algorithm is converted to an optimal algorithm for routing without selective transfers. Each selective transfer routing step, a masked swap of combinatorial rectangles, can be converted into $\bigo{\sqrt N}$ swaps of combinatorial rectangles without a mask by simply performing one swap for each row of the masked rectangles. One may optimally decompose each masked rectangle into a sum of disjoint rectangles as well~\cite{tan_depthoptimal_2024}, but in the worst case this still requires $\bigomega{\sqrt N}$ rectangles, so it does not improve the worst-case performance of the algorithm. This algorithm for routing with grid transfers takes at most $\sqrt N (2 \log_2 N - 1)$ routing steps.

\textit{Sparse routing.---}We have shown a significant separation between the cost of 2D routing with and without selective transfers. This routing advantage does not hold for a class of sparse permutations, defined as follows: A permutation $\sigma_s$ is column(row)-sparse if $\sigma_s$ is non-identity on at most $\poly(\log N)$ elements per column (row) (see Sec.~\ref{supp-sec:sparse} in SM~\footnotemark[1]). Any such row-sparse or column-sparse permutation can be routed in $\poly (\log N)$ steps without selective transfers through an iterative sparse routing procedure. In each iteration for a column-sparse permutation, the procedure chooses two sets of atoms $A$ and $B$, such that $B$ represents the destination location of each atom in $A$. It uses a greedy selection process to ensure that each atom in $A$ and $B$ is in a unique column in its respective set, and that $|A| = |B| = \bigomega{\sqrt N}$. Each iteration performs a swap of each $i \in A$ with its destination $\sigma(i) \in B$, while leaving other atoms undisturbed. It does so by performing a compression step, which recursively swaps the atoms $A$ downwards to their own row, and then the atoms $B$ to an adjacent row. This step disturbs the positions of qubits not in $A$ or $B$. After this, 1D routing is performed on the $A$ row to align each $i \in A$ to its partner $\sigma(i) \in B$ in the same column, and the two rows are swapped. Finally, the compression procedure is inverted, thus completing the swap of each $i \in A$, $\sigma(i) \in B$, and returning all other atoms to their original positions, completing one iteration. The procedure is similar for a row-sparse permutation.

Each compression and its reverse take at most $\log_2 \sqrt N$ steps, as does the 1D routing. In each iteration, $\Omega(\sqrt N)$ atoms are routed to their destinations, and any sparse permutation contains at most $\bigo{\sqrt N \poly(\log \sqrt N)}$ qubits, so $\poly(\log \sqrt N)$ iterations are performed for a total of $\poly(\log N)$ routing steps.

\textit{Discussion and outlook.---}We established a polynomial bound for routing on 2D arrays in state-of-the-art reconfigurable neutral atom array designs. Conversely, we showed that a simple experimental upgrade, \emph{selective transfer} between SLM and AOD grids, reduces the lower bound to $\bigomega{\log N}$ steps, yielding a significant speedup. Using a hypercube embedding for the qubits, we provided protocols that perform any permutation in $\bigo{\log N}$ and $\bigo{\sqrt N \log N}$ steps for routing with and without selective transfers, respectively, saturating the lower bounds.

All three stages of circuit compilation---scheduling Rydberg-based two-qubit gates, mapping qubits to atoms, and routing qubits---must be optimized together to obtain the most efficient implementation.
Compilers such as \cite{tan_compiling_2024} make use of satisfiability modulo theories (SMT) solvers to find optimal solutions to these steps, but since SMT problems are generally NP-hard, they are prohibitively expensive for practical use. Alternatively, compilers such as \cite{tan_compilation_2024_enola, wang_qpilot_2024, wang_fpqac_2023} make use of heuristic solutions, which may be sub-optimal but can be obtained more efficiently. In this work, we have shown lower bounds on the time taken to perform the routing step, as well as an optimal protocol for worst-case permutations, based on routing on the hypercube. However, this protocol may not perform well for all permutations. Our sparse routing protocol may be a useful heuristic routing protocol that applies to a broad class of permutations and routes efficiently. Future research can explore how these various methods can be combined to maximize performance and solver scalability.

Although our results are presented for rectangular arrays, they are relevant to any array that can be considered a sub-array of a bigger rectangular one with only $\bigo{N}$ additional sites, such as hexagonal, triangular, or Kagome lattices. One can reduce the routing problem on the non-rectangular sub-array $A$ to a routing problem on the rectangular array $R$ where the permutation on $R$ is the same as the one on $A$ for those sites and arbitrary for $R\setminus A$ sites. Thus, all our bounds apply to such arrays as well.

We proved the lower bound for routing with grid transfers by a counting argument. It could be illuminating to find a constructive proof for this statement that showcases properties of a permutation that make its implementation harder rather than only showing that hard permutations exist.

\begin{acknowledgments}

\textit{Acknowledgments.---}We thank Daniel Bochen Tan and Juntian Tu for useful discussions. 
A.F., D.D., M.J.G., J.V.P., A.M.C.,~and A.V.G.~were supported in part by NSF (QLCI award No.~OMA-2120757).
A.F., D.D., A.M.C., and A.V.G.~were supported in part by the DoE ASCR Quantum Testbed Pathfinder program (awards No.~DE-SC0019040 and No.~DE-SC0024220) and the DoE ASCR Accelerated Research in Quantum Computing program (awards No.~DE-SC0020312 and No.~DE-SC0025341).
A.F., D.D.,~and A.V.G.~were supported in part by AFOSR MURI, NSF STAQ program, and DARPA SAVaNT ADVENT.
D.D.\ acknowledges support by the NSF Graduate Research Fellowship Program (GRFP) under Grant No.~DGE-1840340, and an LPS Quantum Graduate Fellowship.
D.B.~acknowledges support from the NSF GRFP (grant DGE1745303) and the Fannie and John Hertz Foundation.
J.V.P.~acknowledges partial support from ONR (award No.~N000142212085).
A.V.G.~also acknowledges support from the U.S.~Department of Energy, Office of Science, National Quantum Information Science Research Centers, Quantum Systems Accelerator.

\end{acknowledgments}

\bibliographystyle{unsrtnat}
\bibliography{references}

\begin{thebibliography}{32}
\providecommand{\natexlab}[1]{#1}
\providecommand{\url}[1]{\texttt{#1}}
\expandafter\ifx\csname urlstyle\endcsname\relax
  \providecommand{\doi}[1]{doi: #1}\else
  \providecommand{\doi}{doi: \begingroup \urlstyle{rm}\Url}\fi

\bibitem[Bluvstein et~al.(2021)Bluvstein, Omran, Levine, Keesling, Semeghini,
  Ebadi, Wang, Michailidis, Maskara, Ho, Choi, Serbyn, Greiner, Vuleti{\'c},
  and Lukin]{bluvstein_controlling_2021}
D.~Bluvstein, A.~Omran, H.~Levine, A.~Keesling, G.~Semeghini, S.~Ebadi, T.~T.
  Wang, A.~A. Michailidis, N.~Maskara, W.~W. Ho, S.~Choi, M.~Serbyn,
  M.~Greiner, V.~Vuleti{\'c}, and M.~D. Lukin.
\newblock Controlling quantum many-body dynamics in driven {{Rydberg}} atom
  arrays.
\newblock \emph{Science}, 371\penalty0 (6536):\penalty0 1355--1359, March 2021.
\newblock URL \url{https://www.science.org/doi/full/10.1126/science.abg2530}.

\bibitem[Cong et~al.(2024)Cong, Maskara, Tran, Pichler, Semeghini, Yelin, Choi,
  and Lukin]{cong_enhancing_2024}
Iris Cong, Nishad Maskara, Minh~C. Tran, Hannes Pichler, Giulia Semeghini,
  Susanne~F. Yelin, Soonwon Choi, and Mikhail~D. Lukin.
\newblock Enhancing detection of topological order by local error correction.
\newblock \emph{Nature Communications}, 15\penalty0 (1):\penalty0 1527,
  February 2024.
\newblock URL \url{https://www.nature.com/articles/s41467-024-45584-6}.

\bibitem[Semeghini et~al.(2021)Semeghini, Levine, Keesling, Ebadi, Wang,
  Bluvstein, Verresen, Pichler, Kalinowski, Samajdar, Omran, Sachdev,
  Vishwanath, Greiner, Vuleti{\'c}, and Lukin]{semeghini_probing_2021}
G.~Semeghini, H.~Levine, A.~Keesling, S.~Ebadi, T.~T. Wang, D.~Bluvstein,
  R.~Verresen, H.~Pichler, M.~Kalinowski, R.~Samajdar, A.~Omran, S.~Sachdev,
  A.~Vishwanath, M.~Greiner, V.~Vuleti{\'c}, and M.~D. Lukin.
\newblock Probing topological spin liquids on a programmable quantum simulator.
\newblock \emph{Science}, 374\penalty0 (6572):\penalty0 1242--1247, December
  2021.
\newblock URL \url{https://www.science.org/doi/10.1126/science.abi8794}.

\bibitem[Ebadi et~al.(2021)Ebadi, Wang, Levine, Keesling, Semeghini, Omran,
  Bluvstein, Samajdar, Pichler, Ho, Choi, Sachdev, Greiner, Vuleti{\'c}, and
  Lukin]{ebadi_quantum_2021}
Sepehr Ebadi, Tout~T. Wang, Harry Levine, Alexander Keesling, Giulia Semeghini,
  Ahmed Omran, Dolev Bluvstein, Rhine Samajdar, Hannes Pichler, Wen~Wei Ho,
  Soonwon Choi, Subir Sachdev, Markus Greiner, Vladan Vuleti{\'c}, and
  Mikhail~D. Lukin.
\newblock Quantum phases of matter on a 256-atom programmable quantum
  simulator.
\newblock \emph{Nature}, 595\penalty0 (7866):\penalty0 227--232, July 2021.
\newblock URL \url{https://www.nature.com/articles/s41586-021-03582-4}.

\bibitem[Ebadi et~al.(2022)Ebadi, Keesling, Cain, Wang, Levine, Bluvstein,
  Semeghini, Omran, Liu, Samajdar, Luo, Nash, Gao, Barak, Farhi, Sachdev,
  Gemelke, Zhou, Choi, Pichler, Wang, Greiner, Vuleti{\'c}, and
  Lukin]{ebadi_quantum_2022}
S.~Ebadi, A.~Keesling, M.~Cain, T.~T. Wang, H.~Levine, D.~Bluvstein,
  G.~Semeghini, A.~Omran, J.-G. Liu, R.~Samajdar, X.-Z. Luo, B.~Nash, X.~Gao,
  B.~Barak, E.~Farhi, S.~Sachdev, N.~Gemelke, L.~Zhou, S.~Choi, H.~Pichler,
  S.-T. Wang, M.~Greiner, V.~Vuleti{\'c}, and M.~D. Lukin.
\newblock Quantum optimization of maximum independent set using {{Rydberg}}
  atom arrays.
\newblock \emph{Science}, 376\penalty0 (6598):\penalty0 1209--1215, June 2022.
\newblock URL \url{https://www.science.org/doi/full/10.1126/science.abo6587}.

\bibitem[Hashizume et~al.(2021)Hashizume, Bentsen, Weber, and
  Daley]{hashizume_deterministic_2021}
Tomohiro Hashizume, Gregory~S. Bentsen, Sebastian Weber, and Andrew~J. Daley.
\newblock Deterministic {{Fast Scrambling}} with {{Neutral Atom Arrays}}.
\newblock \emph{Physical Review Letters}, 126\penalty0 (20):\penalty0 200603,
  May 2021.
\newblock URL \url{https://link.aps.org/doi/10.1103/PhysRevLett.126.200603}.

\bibitem[Graham et~al.(2022)Graham, Song, Scott, Poole, Phuttitarn, Jooya,
  Eichler, Jiang, Marra, Grinkemeyer, Kwon, Ebert, Cherek, Lichtman, Gillette,
  Gilbert, Bowman, Ballance, Campbell, Dahl, Crawford, Blunt, Rogers, Noel, and
  Saffman]{graham_multiqubit_2022}
T.~M. Graham, Y.~Song, J.~Scott, C.~Poole, L.~Phuttitarn, K.~Jooya, P.~Eichler,
  X.~Jiang, A.~Marra, B.~Grinkemeyer, M.~Kwon, M.~Ebert, J.~Cherek, M.~T.
  Lichtman, M.~Gillette, J.~Gilbert, D.~Bowman, T.~Ballance, C.~Campbell, E.~D.
  Dahl, O.~Crawford, N.~S. Blunt, B.~Rogers, T.~Noel, and M.~Saffman.
\newblock Multi-qubit entanglement and algorithms on a neutral-atom quantum
  computer.
\newblock \emph{Nature}, 604\penalty0 (7906):\penalty0 457--462, 2022.
\newblock URL \url{https://www.nature.com/articles/s41586-022-04603-6}.

\bibitem[Radnaev et~al.(2024)Radnaev, Chung, Cole, Mason, Ballance, Bedalov,
  Belknap, Berman, Blakely, Bloomfield, Buttler, Campbell, Chopinaud,
  Copenhaver, Dawes, Eubanks, Friss, Garcia, Gilbert, Gillette, Goiporia,
  Gokhale, Goldwin, Goodwin, Graham, Guttormsson, Hickman, Hurtley, Iliev,
  Jones, Jones, Kuper, Lewis, Lichtman, Majdeteimouri, Mason, McMaster, Miles,
  Mitchell, Murphree, {Neff-Mallon}, Oh, Omole, Simon, Pederson, Perlin,
  Reiter, Rines, Romlow, Scott, Stiefvater, Tanner, Tucker, Vinogradov, Warter,
  Yeo, Saffman, and Noel]{radnaev_universal_2024}
A.~G. Radnaev, W.~C. Chung, D.~C. Cole, D.~Mason, T.~G. Ballance, M.~J.
  Bedalov, D.~A. Belknap, M.~R. Berman, M.~Blakely, I.~L. Bloomfield, P.~D.
  Buttler, C.~Campbell, A.~Chopinaud, E.~Copenhaver, M.~K. Dawes, S.~Y.
  Eubanks, A.~J. Friss, D.~M. Garcia, J.~Gilbert, M.~Gillette, P.~Goiporia,
  P.~Gokhale, J.~Goldwin, D.~Goodwin, T.~M. Graham, C.~J. Guttormsson, G.~T.
  Hickman, L.~Hurtley, M.~Iliev, E.~B. Jones, R.~A. Jones, K.~W. Kuper, T.~B.
  Lewis, M.~T. Lichtman, F.~Majdeteimouri, J.~J. Mason, J.~K. McMaster, J.~A.
  Miles, P.~T. Mitchell, J.~D. Murphree, N.~A. {Neff-Mallon}, T.~Oh, V.~Omole,
  C.~Parlo Simon, N.~Pederson, M.~A. Perlin, A.~Reiter, R.~Rines, P.~Romlow,
  A.~M. Scott, D.~Stiefvater, J.~R. Tanner, A.~K. Tucker, I.~V. Vinogradov,
  M.~L. Warter, M.~Yeo, M.~Saffman, and T.~W. Noel.
\newblock A universal neutral-atom quantum computer with individual optical
  addressing and non-destructive readout, August 2024.
\newblock URL \url{http://arxiv.org/abs/2408.08288}.

\bibitem[Browaeys and Lahaye(2020)]{browaeys_manybody_2020}
Antoine Browaeys and Thierry Lahaye.
\newblock Many-body physics with individually controlled {{Rydberg}} atoms.
\newblock \emph{Nature Physics}, 16\penalty0 (2):\penalty0 132--142, February
  2020.
\newblock URL \url{https://www.nature.com/articles/s41567-019-0733-z}.

\bibitem[Bruzewicz et~al.(2019)Bruzewicz, Chiaverini, McConnell, and
  Sage]{bruzewicz_trappedion_2019}
Colin~D. Bruzewicz, John Chiaverini, Robert McConnell, and Jeremy~M. Sage.
\newblock Trapped-ion quantum computing: {{Progress}} and challenges.
\newblock \emph{Applied Physics Reviews}, 6\penalty0 (2):\penalty0 021314,
  2019.
\newblock URL \url{https://doi.org/10.1063/1.5088164}.

\bibitem[Kjaergaard et~al.(2020)Kjaergaard, Schwartz, Braumüller, Krantz,
  Wang, Gustavsson, and Oliver]{kjaergaard_superconducting_2020}
Morten Kjaergaard, Mollie~E. Schwartz, Jochen Braumüller, Philip Krantz, Joel
  I.-J. Wang, Simon Gustavsson, and William~D. Oliver.
\newblock Superconducting {{Qubits}}: {{Current State}} of {{Play}}.
\newblock \emph{Annual Review of Condensed Matter Physics}, 11:\penalty0
  369--395, 2020.
\newblock URL
  \url{https://www.annualreviews.org/content/journals/10.1146/annurev-conmatphys-031119-050605}.

\bibitem[Beugnon et~al.(2007)Beugnon, Tuchendler, Marion, Ga{\"e}tan,
  Miroshnychenko, Sortais, Lance, Jones, Messin, Browaeys, and
  Grangier]{beugnon_twodimensional_2007}
J{\'e}r{\^o}me Beugnon, Charles Tuchendler, Harold Marion, Alpha Ga{\"e}tan,
  Yevhen Miroshnychenko, Yvan R.~P. Sortais, Andrew~M. Lance, Matthew P.~A.
  Jones, Ga{\'e}tan Messin, Antoine Browaeys, and Philippe Grangier.
\newblock Two-dimensional transport and transfer of a single atomic qubit in
  optical tweezers.
\newblock \emph{Nature Physics}, 3\penalty0 (10):\penalty0 696--699, October
  2007.
\newblock URL \url{https://www.nature.com/articles/nphys698}.

\bibitem[Bluvstein et~al.(2022)Bluvstein, Levine, Semeghini, Wang, Ebadi,
  Kalinowski, Keesling, Maskara, Pichler, Greiner, Vuleti{\'c}, and
  Lukin]{bluvstein_quantum_2022}
Dolev Bluvstein, Harry Levine, Giulia Semeghini, Tout~T. Wang, Sepehr Ebadi,
  Marcin Kalinowski, Alexander Keesling, Nishad Maskara, Hannes Pichler, Markus
  Greiner, Vladan Vuleti{\'c}, and Mikhail~D. Lukin.
\newblock A quantum processor based on coherent transport of entangled atom
  arrays.
\newblock \emph{Nature}, 604\penalty0 (7906):\penalty0 451--456, April 2022.
\newblock URL \url{https://www.nature.com/articles/s41586-022-04592-6}.

\bibitem[Bluvstein et~al.(2024)Bluvstein, Evered, Geim, Li, Zhou, Manovitz,
  Ebadi, Cain, Kalinowski, Hangleiter, Bonilla~Ataides, Maskara, Cong, Gao,
  Sales~Rodriguez, Karolyshyn, Semeghini, Gullans, Greiner, Vuleti{\'c}, and
  Lukin]{bluvstein_logical_2024}
Dolev Bluvstein, Simon~J. Evered, Alexandra~A. Geim, Sophie~H. Li, Hengyun
  Zhou, Tom Manovitz, Sepehr Ebadi, Madelyn Cain, Marcin Kalinowski, Dominik
  Hangleiter, J.~Pablo Bonilla~Ataides, Nishad Maskara, Iris Cong, Xun Gao,
  Pedro Sales~Rodriguez, Thomas Karolyshyn, Giulia Semeghini, Michael~J.
  Gullans, Markus Greiner, Vladan Vuleti{\'c}, and Mikhail~D. Lukin.
\newblock Logical quantum processor based on reconfigurable atom arrays.
\newblock \emph{Nature}, 626\penalty0 (7997):\penalty0 58--65, February 2024.
\newblock URL \url{https://www.nature.com/articles/s41586-023-06927-3}.

\bibitem[Xu et~al.(2024)Xu, Bonilla~Ataides, Pattison, Raveendran, Bluvstein,
  Wurtz, Vasi{\'c}, Lukin, Jiang, and Zhou]{xu_constantoverhead_2024}
Qian Xu, J.~Pablo Bonilla~Ataides, Christopher~A. Pattison, Nithin Raveendran,
  Dolev Bluvstein, Jonathan Wurtz, Bane Vasi{\'c}, Mikhail~D. Lukin, Liang
  Jiang, and Hengyun Zhou.
\newblock Constant-overhead fault-tolerant quantum computation with
  reconfigurable atom arrays.
\newblock \emph{Nature Physics}, 20\penalty0 (7):\penalty0 1084--1090, July
  2024.
\newblock URL \url{https://www.nature.com/articles/s41567-024-02479-z}.

\bibitem[Wang et~al.(2024)Wang, Tan, Liu, Liu, Gu, Cong, and
  Han]{wang_qpilot_2024}
Hanrui Wang, Daniel~Bochen Tan, Pengyu Liu, Yilian Liu, Jiaqi Gu, Jason Cong,
  and Song Han.
\newblock Q-{{Pilot}}: {{Field Programmable Qubit Array Compilation}} with
  {{Flying Ancillas}}, September 2024.
\newblock URL \url{http://arxiv.org/abs/2311.16190}.

\bibitem[Tan et~al.(2024{\natexlab{a}})Tan, Bluvstein, Lukin, and
  Cong]{tan_compiling_2024}
Daniel~Bochen Tan, Dolev Bluvstein, Mikhail~D. Lukin, and Jason Cong.
\newblock Compiling {{Quantum Circuits}} for {{Dynamically Field-Programmable
  Neutral Atoms Array Processors}}.
\newblock \emph{Quantum}, 8:\penalty0 1281, March 2024{\natexlab{a}}.
\newblock URL \url{https://quantum-journal.org/papers/q-2024-03-14-1281/}.

\bibitem[Patel et~al.(2022)Patel, Silver, and Tiwari]{patel_geyser_2022a}
Tirthak Patel, Daniel Silver, and Devesh Tiwari.
\newblock Geyser: A compilation framework for quantum computing with neutral
  atoms.
\newblock In \emph{Proceedings of the 49th {{Annual International Symposium}}
  on {{Computer Architecture}}}, pages 383--395, New York New York, June 2022.
  ACM.
\newblock URL \url{https://dl.acm.org/doi/10.1145/3470496.3527428}.

\bibitem[Patel et~al.(2023)Patel, Silver, and Tiwari]{patel_graphine_2023}
Tirthak Patel, Daniel Silver, and Devesh Tiwari.
\newblock {{GRAPHINE}}: {{Enhanced Neutral Atom Quantum Computing}} using
  {{Application-Specific Rydberg Atom Arrangement}}.
\newblock In \emph{Proceedings of the {{International Conference}} for {{High
  Performance Computing}}, {{Networking}}, {{Storage}} and {{Analysis}}}, pages
  1--15, Denver CO USA, November 2023. ACM.
\newblock URL \url{https://dl.acm.org/doi/10.1145/3581784.3607032}.

\bibitem[Ludmir and Patel(2024)]{ludmir_parallax_2024}
Jason Ludmir and Tirthak Patel.
\newblock Parallax: {{A Compiler}} for {{Neutral Atom Quantum Computers}} under
  {{Hardware Constraints}}, September 2024.
\newblock URL \url{http://arxiv.org/abs/2409.04578}.

\bibitem[Tan et~al.(2024{\natexlab{b}})Tan, Lin, and
  Cong]{tan_compilation_2024_enola}
Daniel~Bochen Tan, Wan-Hsuan Lin, and Jason Cong.
\newblock Compilation for {{Dynamically Field-Programmable Qubit Arrays}} with
  {{Efficient}} and {{Provably Near-Optimal Scheduling}}, November
  2024{\natexlab{b}}.
\newblock URL \url{http://arxiv.org/abs/2405.15095}.

\bibitem[Cowtan et~al.(2019)Cowtan, Dilkes, Duncan, Krajenbrink, Simmons, and
  Sivarajah]{cowtan_qubit_2019}
Alexander Cowtan, Silas Dilkes, Ross Duncan, Alexandre Krajenbrink, Will
  Simmons, and Seyon Sivarajah.
\newblock On the {{Qubit Routing Problem}}.
\newblock In \emph{14th {{Conference}} on the {{Theory}} of {{Quantum
  Computation}}, {{Communication}} and {{Cryptography}} ({{TQC}} 2019)}, pages
  5:1--5:32. Schloss Dagstuhl -- Leibniz-Zentrum f{\"u}r Informatik, 2019.
\newblock URL
  \url{https://drops.dagstuhl.de/entities/document/10.4230/LIPIcs.TQC.2019.5}.

\bibitem[Childs et~al.(2019)Childs, Schoute, and Unsal]{childs_circuit_2019}
Andrew~M. Childs, Eddie Schoute, and Cem~M. Unsal.
\newblock Circuit transformations for quantum architectures.
\newblock In \emph{14th Conference on the Theory of Quantum Computation,
  Communication and Cryptography (TQC 2019)}, pages 3:1--3:24. Schloss Dagstuhl
  -- Leibniz-Zentrum f{\"u}r Informatik, 2019.
\newblock \doi{10.4230/LIPIcs.TQC.2019.3}.
\newblock URL
  \url{https://drops.dagstuhl.de/entities/document/10.4230/LIPIcs.TQC.2019.3}.

\bibitem[Bapat et~al.(2021)Bapat, Childs, Gorshkov, King, Schoute, and
  Shastri]{bapat_quantum_2021}
Aniruddha Bapat, Andrew~M. Childs, Alexey~V. Gorshkov, Samuel King, Eddie
  Schoute, and Hrishee Shastri.
\newblock Quantum routing with fast reversals.
\newblock \emph{Quantum}, 5:\penalty0 533, August 2021.
\newblock URL \url{http://arxiv.org/abs/2103.03264}.

\bibitem[Yuan and Zhang(2024)]{yuan_full_2024}
Pei Yuan and Shengyu Zhang.
\newblock Full {{Characterization}} of the {{Depth Overhead}} for {{Quantum
  Circuit Compilation}} with {{Arbitrary Qubit Connectivity Constraint}},
  February 2024.
\newblock URL \url{http://arxiv.org/abs/2402.02403}.

\bibitem[Bergamini et~al.(2004)Bergamini, Darqui{\'e}, Jones, Jacubowiez,
  Browaeys, and Grangier]{bergamini_holographic_2004}
Silvia Bergamini, Beno{\^i}t Darqui{\'e}, Matthew Jones, Lionel Jacubowiez,
  Antoine Browaeys, and Philippe Grangier.
\newblock Holographic generation of microtrap arrays for single atoms by use of
  a programmable phase modulator.
\newblock \emph{JOSA B}, 21\penalty0 (11):\penalty0 1889--1894, November 2004.
\newblock URL
  \url{https://opg.optica.org/josab/abstract.cfm?uri=josab-21-11-1889}.

\bibitem[Barredo et~al.(2016)Barredo, {de L{\'e}s{\'e}leuc}, Lienhard, Lahaye,
  and Browaeys]{barredo_atomatom_2016}
Daniel Barredo, Sylvain {de L{\'e}s{\'e}leuc}, Vincent Lienhard, Thierry
  Lahaye, and Antoine Browaeys.
\newblock An atom-by-atom assembler of defect-free arbitrary two-dimensional
  atomic arrays.
\newblock \emph{Science}, 354\penalty0 (6315):\penalty0 1021--1023, November
  2016.
\newblock URL \url{https://www.science.org/doi/full/10.1126/science.aah3778}.

\bibitem[Barredo et~al.(2018)Barredo, Lienhard, {de L{\'e}s{\'e}leuc}, Lahaye,
  and Browaeys]{barredo_synthetic_2018}
Daniel Barredo, Vincent Lienhard, Sylvain {de L{\'e}s{\'e}leuc}, Thierry
  Lahaye, and Antoine Browaeys.
\newblock Synthetic three-dimensional atomic structures assembled atom by atom.
\newblock \emph{Nature}, 561\penalty0 (7721):\penalty0 79--82, September 2018.
\newblock URL \url{https://www.nature.com/articles/s41586-018-0450-2}.

\bibitem[Note1()]{Note1}
Note1.
\newblock See Supplemental Material for more details about our models, formal
  definitions, detailed proofs of theorems, as well as algorithms for hypercube
  routing and sparse routing.

\bibitem[Alon et~al.(1993)Alon, Chung, and Graham]{alon_routing_1993}
N.~Alon, F.~R.~K. Chung, and R.~L. Graham.
\newblock Routing permutations on graphs via matchings.
\newblock In \emph{Proceedings of the Twenty-Fifth Annual {{ACM}} Symposium on
  {{Theory}} of Computing - {{STOC}} '93}, pages 583--591, San Diego,
  California, United States, 1993. ACM Press.
\newblock URL \url{http://portal.acm.org/citation.cfm?doid=167088.167239}.

\bibitem[Tan et~al.(2024{\natexlab{c}})Tan, Ping, and
  Cong]{tan_depthoptimal_2024}
Daniel~Bochen Tan, Shuohao Ping, and Jason Cong.
\newblock Depth-{{Optimal Addressing}} of {{2D Qubit Array}} with {{1D Controls
  Based}} on {{Exact Binary Matrix Factorization}}.
\newblock In \emph{2024 {{Design}}, {{Automation}} \& {{Test}} in {{Europe
  Conference}} \& {{Exhibition}} ({{DATE}})}, pages 1--6, March
  2024{\natexlab{c}}.
\newblock URL
  \url{https://ieeexplore.ieee.org/abstract/document/10546763?casa_token=lGPEuN7-7lYAAAAA:xdx2w8teKbuKaOb_pN0MSlM8e4LNEipYUcEJ6rfsf9g38lWLsxKEEDlCV_xJ_VBZmHeH9V_dE5w}.

\bibitem[Wang et~al.(2023)Wang, Liu, Tan, Liu, Gu, Pan, Cong, Acar, and
  Han]{wang_fpqac_2023}
Hanrui Wang, Pengyu Liu, Bochen Tan, Yilian Liu, Jiaqi Gu, David~Z. Pan, Jason
  Cong, Umut Acar, and Song Han.
\newblock {{FPQA-C}}: {{A Compilation Framework}} for {{Field Programmable
  Qubit Array}}, November 2023.
\newblock URL \url{http://arxiv.org/abs/2311.15123}.

\end{thebibliography}


\begin{thebibliography}{5}
\providecommand{\natexlab}[1]{#1}
\providecommand{\url}[1]{\texttt{#1}}
\expandafter\ifx\csname urlstyle\endcsname\relax
  \providecommand{\doi}[1]{doi: #1}\else
  \providecommand{\doi}{doi: \begingroup \urlstyle{rm}\Url}\fi

\bibitem[Bluvstein et~al.(2022)Bluvstein, Levine, Semeghini, Wang, Ebadi,
  Kalinowski, Keesling, Maskara, Pichler, Greiner, Vuleti{\'c}, and
  Lukin]{bluvstein_quantum_2022}
Dolev Bluvstein, Harry Levine, Giulia Semeghini, Tout~T. Wang, Sepehr Ebadi,
  Marcin Kalinowski, Alexander Keesling, Nishad Maskara, Hannes Pichler, Markus
  Greiner, Vladan Vuleti{\'c}, and Mikhail~D. Lukin.
\newblock A quantum processor based on coherent transport of entangled atom
  arrays.
\newblock \emph{Nature}, 604\penalty0 (7906):\penalty0 451--456, April 2022.
\newblock URL \url{https://www.nature.com/articles/s41586-022-04592-6}.

\bibitem[Bluvstein et~al.(2024)Bluvstein, Evered, Geim, Li, Zhou, Manovitz,
  Ebadi, Cain, Kalinowski, Hangleiter, Bonilla~Ataides, Maskara, Cong, Gao,
  Sales~Rodriguez, Karolyshyn, Semeghini, Gullans, Greiner, Vuleti{\'c}, and
  Lukin]{bluvstein_logical_2024}
Dolev Bluvstein, Simon~J. Evered, Alexandra~A. Geim, Sophie~H. Li, Hengyun
  Zhou, Tom Manovitz, Sepehr Ebadi, Madelyn Cain, Marcin Kalinowski, Dominik
  Hangleiter, J.~Pablo Bonilla~Ataides, Nishad Maskara, Iris Cong, Xun Gao,
  Pedro Sales~Rodriguez, Thomas Karolyshyn, Giulia Semeghini, Michael~J.
  Gullans, Markus Greiner, Vladan Vuleti{\'c}, and Mikhail~D. Lukin.
\newblock Logical quantum processor based on reconfigurable atom arrays.
\newblock \emph{Nature}, 626\penalty0 (7997):\penalty0 58--65, February 2024.
\newblock URL \url{https://www.nature.com/articles/s41586-023-06927-3}.

\bibitem[Xu et~al.(2024)Xu, Bonilla~Ataides, Pattison, Raveendran, Bluvstein,
  Wurtz, Vasi{\'c}, Lukin, Jiang, and Zhou]{xu_constantoverhead_2024}
Qian Xu, J.~Pablo Bonilla~Ataides, Christopher~A. Pattison, Nithin Raveendran,
  Dolev Bluvstein, Jonathan Wurtz, Bane Vasi{\'c}, Mikhail~D. Lukin, Liang
  Jiang, and Hengyun Zhou.
\newblock Constant-overhead fault-tolerant quantum computation with
  reconfigurable atom arrays.
\newblock \emph{Nature Physics}, 20\penalty0 (7):\penalty0 1084--1090, July
  2024.
\newblock URL \url{https://www.nature.com/articles/s41567-024-02479-z}.

\bibitem[Tan et~al.(2024)Tan, Ping, and Cong]{tan_depthoptimal_2024}
Daniel~Bochen Tan, Shuohao Ping, and Jason Cong.
\newblock Depth-{{Optimal Addressing}} of {{2D Qubit Array}} with {{1D Controls
  Based}} on {{Exact Binary Matrix Factorization}}.
\newblock In \emph{2024 {{Design}}, {{Automation}} \& {{Test}} in {{Europe
  Conference}} \& {{Exhibition}} ({{DATE}})}, pages 1--6, March 2024.
\newblock URL
  \url{https://ieeexplore.ieee.org/abstract/document/10546763?casa_token=lGPEuN7-7lYAAAAA:xdx2w8teKbuKaOb_pN0MSlM8e4LNEipYUcEJ6rfsf9g38lWLsxKEEDlCV_xJ_VBZmHeH9V_dE5w}.

\bibitem[Alon et~al.(1993)Alon, Chung, and Graham]{alon_routing_1993}
N.~Alon, F.~R.~K. Chung, and R.~L. Graham.
\newblock Routing permutations on graphs via matchings.
\newblock In \emph{Proceedings of the Twenty-Fifth Annual {{ACM}} Symposium on
  {{Theory}} of Computing - {{STOC}} '93}, pages 583--591, San Diego,
  California, United States, 1993. ACM Press.
\newblock URL \url{http://portal.acm.org/citation.cfm?doid=167088.167239}.

\end{thebibliography}

\makeatletter\@input{xx_supp.tex}\makeatother
\end{document}


\title{Supplemental Material}

\maketitle

\else %
\newcommand{\refmain}[1]{\ref{#1}}
\fi %

\tableofcontents

\section{Routing Models}
In this section, we formulate and study multiple models of routing motivated by current experimental capabilities as well as a feasible upgrade. Each model is specified by a layout geometry in which the qubits are placed (a 1D chain or 2D grid), and a set of permutations which are considered single routing steps. A routing problem is then defined by a set of qubits $S$ with $|S| = N$, a permutation $\sigma\colon S \to S$ specifying the target configuration of the qubits, and a set of allowed single-step permutations $M$. %

\begin{definition}[Routing number]\label{def:routingnumber}
    For a given model $\alpha$ with $N$ atoms and a set of allowed single-step permutations $M$, a routing sequence is a list of sequential operations, called routing steps, $[m_1, m_2, \dots m_l]$, where each step $m_i$ is from the allowed set of single-step permutations $M$. The routing number $\Rt^{\alpha}_{\sigma}(N)$ is the minimum number of routing steps to implement the permutation $\sigma$.

    We also define the worst-case routing number 
    \begin{equation}
        \Rt^{\alpha} (N) = \max_{\sigma} \Rt_{\sigma}^{\alpha}(N).\label{eq:rtalpha}
    \end{equation}
    (In some cases we omit the label $\alpha$ when the model is clear from context.)
\end{definition}

The models we consider allow significant non-locality: moves involving the transportation of qubits across the entire extent of the qubit array may be considered single steps, even though 
the current experimental time cost of transporting a physical atom a distance $d$ scales as %
$\sqrt{d}$ in the far-distance limit \cite{bluvstein_quantum_2022}. We make this assumption because, at the scale of recent experiments, free-space transport times are typically comparable to 
the time required to transfer atoms between AOD and SLM traps, with both being on the order of $\sim 200 \mu s$ \cite{bluvstein_logical_2024}. In addition, in practice there is negligible impact of transport on qubit fidelity, but infidelities near $\lesssim 0.1 \%$ are %
caused by the transfer of atoms between tweezers. This further justifies our choice of single moves as being operations that involve a fixed number of AOD-SLM transfers, as these should incur a roughly constant level of error, regardless of the distance atoms are transported in a given operation \cite{bluvstein_quantum_2022,bluvstein_logical_2024}. Finally, our analyses apply generally to logical qubits formed by blocks of atoms or single-atom physical qubits, given the rectangular structure of AODs.

We describe two models of routing qubits stored in a one-dimensional line of static optical tweezers. In the first, each single step is a riffle shuffle permutation. Each step is implemented by picking up a subset $A$ of the array and interleaving it in between other atoms in the array. 

\begin{definition}[Single riffle shuffle routing step]\label{def:1DInsStep}
    For a one-dimensional array of $N$ qubits with relative positions $S = \{0, 1, \dots, N - 1\}$, a single routing step is a permutation $\sigma\colon S \to S$ given by $A \subset S$ and %
    functions $f\colon A \to S$, $g\colon S \setminus A \to S$, such that %
    \begin{equation}
        \sigma(i) = 
        \begin{cases}
            f(i) & \text{if $i \in A$}\\
            g(i) & \text{otherwise}.
        \end{cases}
    \end{equation}
    The functions $f$ and $g$ are constrained to be strictly increasing over their domains:
    \begin{align}
        \forall i,j \in A, i < j &\implies f(i) < f(j)\\
        \forall i,j \in S \setminus A, i < j &\implies g(i) < g(j).
    \end{align}

\end{definition}

We depict an example implementation of a single riffle shuffle step in Fig.~\ref{fig:1dInsStep}. With a 1D lattice of sufficiently spaced SLM traps, one performs a riffle shuffle by transporting the qubits in $A$ to appropriate positions between the lattice, and then picking up the rest of the array and expanding it by steering the AOD columns independently to place all of the atoms in SLM traps again. It should be noted that the transport of occupied AOD traps over empty SLM traps is avoided in order to maximize transport fidelity. %
Alternatively, Ref.~\cite{xu_constantoverhead_2024} describes another possible implementation: one generates an additional set of $N/2$ empty static tweezers, and then expands the whole array (again, by spacing out the AOD columns) to leave vacancies where one would like to insert $A$, transports $A$ into these vacancies, and then expands or contracts the array to prepare for the next routing step.

\begin{figure}
    \centering
    \includegraphics[width=250pt]{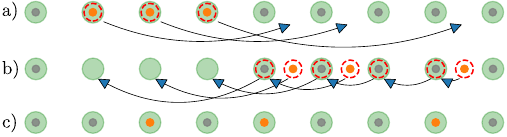}
    \caption{An example of a single riffle shuffle routing step. Light green disks and dashed red circles are SLM- and AOD-generated traps, respectively. Smaller orange dots are the atoms in the subset A (see Def.~\ref{def:1DInsStep}) and gray dots are the rest of the atoms. %
    \textbf{(a)} Orange atoms are picked up, moved down, and then transported in between their target locations. \textbf{(b)} Additional AOD columns are ramped up, picking up the rest of the %
    out-of-place atoms. All atoms in the AOD traps are then transported downwards, expanded, and then put in place in the SLM traps. \textbf{(c)} Final ordering.} 
    \label{fig:1dInsStep}
\end{figure}

We define an alternative form of 1D routing in which a single step is an in-order swap between subsets $A$ and $B$ of the array. Each single step of this model can be implemented by two riffle shuffles, but we do not know if every riffle shuffle can be implemented in a constant number of in-order swap steps. We show later that routing with riffle shuffles achieves a constant factor speedup of $\log(3) / \log(2)$ over in-order swaps when reversing the order of all qubits. Notably, with in-order swaps, only the atoms being swapped are transferred between traps, whereas with riffle shuffles one may be required to transfer every atom of the array between SLM and AOD traps to insert even a single qubit into a different location in the array.

\begin{definition}[Single in-order swap routing step]\label{def:1dSwapStep}
    For a row of atoms whose positions are labeled by $S = \{0, 1, \dots, N - 1 \}$, one single swap is a permutation $\sigma\colon S \to S$ for which $\exists A, B \subseteq S$, $A = \{a_1, a_2, \dots, a_k\}, B = \{b_1, b_2, \dots, b_k\}$, $A \cap B = \emptyset$, $a_1 < a_2 < \dots < a_k$, $b_1 < b_2 < \dots < b_k$ and 
    \begin{equation}
        \sigma(i) = \begin{cases}
            b_i & \textrm{if\ } i \in A\\
            a_i & \textrm{if\ } i \in B\\
            i & \textrm{otherwise}.
        \end{cases}
    \end{equation}
\end{definition}

In Fig.\ \ref{fig:1dSwapStep}, we show an example of a rearrangement-based %
implementation of a single step in this model, where rearrangement-based means utilizing only AOD-SLM transfers and shuttling. %
One could alternatively implement single steps in this model via parallel swap gates, forming pairs of qubits in $A$ with those in $B$ by transport followed by a combination of single qubit and CZ gates between the pairs of qubits $A$ and $B$. This approach would likely introduce correlated errors not present in the rearrangement approach, however.%
\begin{figure}
    \centering
    \includegraphics[width=410pt]{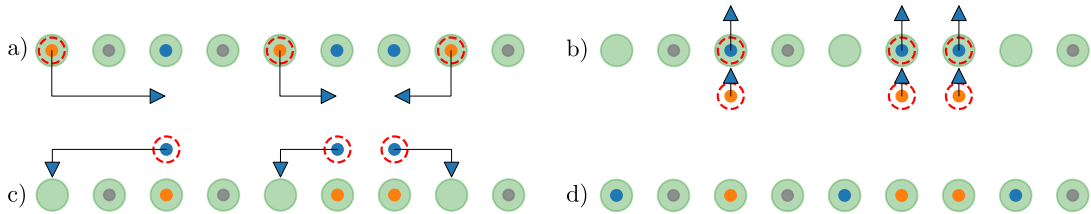}
    \caption{An example of a single step with 1D in-order swaps. \textbf{(a)} Orange qubits are picked up, moved down and then transported beneath their target locations. \textbf{(b)} Additional AOD columns are ramped up, picking up the blue qubits. Qubits are then transported upwards, replacing blue with orange. \textbf{(c)} One row of the AODs is ramped down, depositing the blue qubits, and then orange qubits are transported to their destinations. \textbf{(d)} Final ordering.}
    \label{fig:1dSwapStep}
\end{figure}

We next formulate models of routing in two dimensions. In our first model, one only performs grid transfers of atoms, and a single step is a swap of any two \emph{rectangles} (subgrids) of the same dimension in a 2D grid of atoms.

\begin{definition}[Combinatorial Rectangle]\label{def:rect}
    A combinatorial rectangle $R$ is a set of the form $A \times B = \{ (i,j) \mid i \in A, j \in B \}$, where $A, B \subset \mathbb{Z}$ are the rows and columns of the rectangle. We refer to the $i$th element of a rectangle $R$ as $[R]_i$, where the points are ordered lexicographically, first by row, and then by column [i.e., $R = \{1, 3\} \times \{1, 3\}$ is ordered $R = ( (1, 1), (1,3), (3, 1), (3,3) )$]. %
    The dimension of a rectangle $R$, $\dim(R)$, is the ordered pair $(|A|, |B|)$.
\end{definition}

\begin{definition}[Single routing step with grid transfers]\label{def:2DRectSwaps}
    For $N = m m^\prime$ atoms on an $m \times m^\prime$ grid of static traps, a single step of routing is an in-order swap of two disjoint rectangles, represented by a permutation $\sigma$ such that there exist rectangles $R_1, R_2$ of the same dimension, with $R_1$ and $R_2$ sharing no points, where
    \begin{equation}
        \sigma(j) = \begin{cases}
            [R_2]_i & \text{\ if $j=[R_1]_i$ for some $i$}\\
            [R_1]_i & \text{\ if $j=[R_2]_i$ for some $i$}\\
            j & \text{otherwise}.
        \end{cases}
    \end{equation}
\end{definition}

Figure~\ref{fig:RectSwapEx} shows the steps to perform a swap between rectangles by rearranging the atoms. %
One uses grid transfers to pick up a rectangle $R_1$ and translate its rows and columns to be placed half of a lattice constant to one side of %
an equal-dimension rectangle $R_2$ [panel (b)], then picks up $R_2$ [panel (c)], translates %
both $R_1$ and $R_2$ by half of a lattice constant such that $R_1$ goes to the former location of $R_2$ and deposits it there
[panel (d)], and then translates and deposits $R_2$ in the former location of $R_1$ [panel (e)]. 
We require $R_1$ and $R_2$ to have equal dimension, but not necessarily the same shape, as one may deform the array $R$ in transport. Finally, our definition prevents the reordering of columns or rows in any of the swapped rectangles $R_1$ or $R_2$, as crossing AOD columns or rows would induce atom collisions, and potentially damage the AOD hardware, as mentioned in the main text. %
Alternatively, we could implement these single-step swaps via entangling gates analogously to 1D in-order swaps: transport the rectangle $R_1$ to form nearest-neighbor pairs with those qubits in $R_2$, perform  combinations %
of CZ gates and single-qubit gates to implement a swap, and finally transport the rectangle $R_1$ back to its original place. However, this again introduces the drawback of correlated noise.

\begin{figure}
    \centering
    \includegraphics[width=500pt]{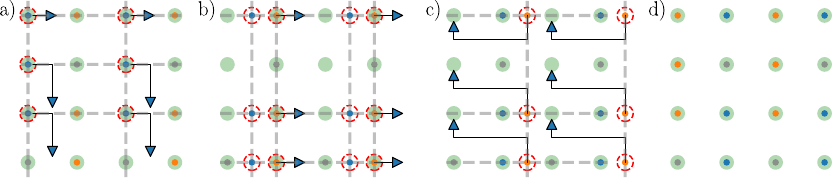}
    \caption{Step-by-step visualization of a swap of two combinatorial rectangles of a 2D array. Orange and blue dots are the two subsets involved in the swap, and gray dots are the rest of the atoms. Green circles are the SLM traps that hold the atoms before and after the swap. Red dashed circles are AOD traps that are used to transport and swap the atoms. \textbf{(a)} First, the first subset of atoms involved in the swap, $A$, are picked up from SLM traps by the AOD traps and are transported next to the second subset of atoms, $B$. \textbf{(b)} Then, using an additional row of AOD traps, $A$ ($B$) atoms are moved into (out of) the SLM traps. \textbf{(c)} Next, $B$ atoms are transported back to the original SLM traps occupied by $A$ atoms. \textbf{(d)} The final arrangement of the atoms after the swap. %
    \label{fig:RectSwapEx}}
\end{figure}

Finally, we consider a more powerful %
model of routing, extending the capabilities of grid transfers, in which one utilizes selective transfers to swap less constrained subsets of the 2D array. %

\begin{definition}[Single routing step with selective transfers]\label{def:2DArbPickStep}
    For $N$ atoms whose positions are labeled by $S^{2D} \subseteq \mathbb{Z} \times \mathbb{Z}$, a single step of routing is a permutation $\sigma\colon S^{2D} \to S^{2D}$ given by two equal-dimension rectangles $R_1, R_2$ and a masking function $\mathcal M \colon \mathbb{Z} \to \{0, 1\}$, such that %
    \begin{equation}
        \sigma(j) = \begin{cases}
            [R_1]_i & \textrm{if\ } j = [R_2]_i \textrm{\ and\ } \mathcal{M}(i) = 1 \textrm{\ for some\ } i\\
            [R_2]_i & \textrm{if\ } j = [R_1]_i \textrm{\ and\ } \mathcal{M}(i) = 1\textrm{\ for some\ } i\\
            j & \textrm{otherwise}.
        \end{cases}
    \end{equation}
\end{definition}

Though this model has expanded capabilities beyond routing with grid transfers, the other constraints associated with AODs remain: only the columns and rows of the array can be steered independently, and no two columns or rows may cross. The primitive operation of routing in this model is similar to routing with grid transfers, except that one can selectively choose which pairs of atoms are swapped between two rectangles with the masking function $\mathcal M$.

\section{Results for Routing in 1D}

In this section, we analyze one-dimensional routing. We first show algorithms for routing and upper bound the routing number of both one-dimensional models. We then show lower bounds on the routing number for both models via two proof techniques, one of which returns a stronger constant factor. 

\subsection{Protocols for 1D Routing}\label{sec:1d_protocols}

In this subsection, we discuss protocols for 1D routing. \textcite{xu_constantoverhead_2024} present an algorithm for routing with riffle shuffles. Their procedure works by partitioning qubits into the left and right parts of the array according to each qubit's destination side. This can be done with a single shuffle, picking up all qubits targeted towards the right half of the array and inserting them at the far-right end (i.e., those whose labels are in the set $\{i \mid i \leq \lfloor N / 2 - 1 \rfloor, \sigma(i) > \lfloor N / 2 - 1 \rfloor\}$). The algorithm then works recursively on the left and right partitions of the array, further dividing each in half so that each quarter of the whole array contains only qubits being routed to that quarter, etc., until all of the qubits are in their target locations. Each successive partitioning of the array remains a single shuffle step, as each interleaves atoms in order without crossing over into any other partitions of the array. The recurrence for the number of insertion steps is thus $T(N) \leq T(n/2) + 1$, which results in $T(N) \leq \log_2(N)$ steps. This upper-bounds the routing number for routing with shuffles as 
$\Rt(N) \leq \log_2(N)$. %
We demonstrate below a matching lower bound on this model, of $\Rt(N) \geq \log_2(N)$ steps, hence showing their work is optimal and that $\Rt(N) = \log_2(N)$.

For 1D routing with in-order swaps, one can make a small alteration to the routing procedure from Ref.~\cite{xu_constantoverhead_2024}: simply replace the partition step with a swap between the qubits on the left half targeting the right half. The number of atoms that are in the wrong partition on each side of the array must be the same, so this swap may always be performed. Likewise, one recurses on the left and right partitions, continuing to partition the array into quarters, etc., and each successive partitioning remains parallel as they all address distinct regions of the array. This leads to the same recurrence for the number of steps, and $T(N) \leq \log_2 N$. We demonstrate below a lower bound on this model of $\Rt(N) \geq \log_3(N)$ steps, determining the routing number to within a constant factor: $\log_3(N) \leq \Rt(N) \leq \log_2(N)$.

\subsection{Lower bounds for 1D Routing}\label{sec:1d_bounds}

To demonstrate lower bounds for 1D routing models, we define a monotone set $\mathcal R(\sigma)$ related to the reversal permutation. We then show that the size $|\mathcal R(\sigma)|$ of this monotone set can only decrease by a constant factor at each step when routing any permutation, even with an arbitrary amount of extra empty static traps or ancilla qubits.

\begin{definition}[1D Reversal Monotone Set]
    For a permutation $\sigma\colon S \to S$ over $S = \{0, 1, \dots, N - 1\}$, the reversal set $\mathcal R(\sigma)$ is the largest subset $x \subseteq S$ such that
    \begin{equation}
        \forall i,j \in x,\, i < j \implies \sigma(i) > \sigma(j).
    \end{equation}
    If there are multiple $x$ that satisfy this condition, then we choose the lexicographically first such set.

\end{definition}

The reversal set $\mathcal R(\sigma)$ is thus a largest subset of qubits that are reversed with respect to one another in the permutation $\sigma$. For example, in the permutation $\sigma = \bigl(\begin{smallmatrix}
    0 & \underline{1} & 2 & \underline{3} & \underline{4}\\
    1 & 4 & 0 & 3 & 2 
\end{smallmatrix} \bigr)$, $\mathcal R(\sigma) = \{1,3,4\}$. In this notation, elements $i,j$ at the top and bottom of a column, respectively, indicate that $\sigma$ sends the $i$th atom in the array to the $j$th location. Underlines indicate the elements of the domain in the reversal set. %
For the reversal permutation $\sigma_r(i) = N - 1 - i$, it is clear that $|\mathcal R(\sigma_r)| = N$, and that for any permutation $\sigma$, $1 \leq |\mathcal R(\sigma)| \leq N$.

We now show that the size of this monotone set can only decrease by a constant factor at each step when implementing any permutation in 1D. For each model, consider a permutation $\sigma\colon S \to S$ and a routing schedule of permutations $\sigma_1, \sigma_2, \dots, \sigma_k$, where $\sigma_i$ is the $i$th permutation applied and is an allowed single step in the relevant model, %
and implement $\sigma$ such that $\sigma = \sigma_k \sigma_{k-1} \dots \sigma_1$.
These steps are ordered swaps or riffle shuffles, depending on the chosen routing model. Before the $i$th step, we let $\tau_i$ denote the remaining permutation to be executed to implement the goal permutation $\sigma$, or $\tau_i = \prod_{j=k}^{i} \sigma_j$. We claim that for routing with riffle shuffles, $|\mathcal R(\tau_{i+1})| \geq \frac{1}{2}|\mathcal R(\tau_i)|$, and with in-order swaps, $|\mathcal R(\tau_{i+1})| \geq \frac{1}{3}|\mathcal R(\tau_i)|$. 

\begin{theorem} %
When routing with riffle shuffles, the size of the largest reversal in a routing schedule can only decrease by a factor of $\frac{1}{2}$ at each step, or $|\mathcal R(\tau_{i+1})| \geq \frac{1}{2} |\mathcal R(\tau_{i})|$.  %
\end{theorem}

\begin{proof}
    As $\tau_{i+1}$ and $\tau_{i}$ are related by a riffle-shuffle permutation, $\tau_i = \tau_{i+1} \sigma_i$, there are two possible cases.

    \textbf{Case 1:} In the $i$th step $\sigma_i$, $j \geq \frac{1}{2}|\mathcal R(\tau_i)|$ qubits from the set $\mathcal R(\tau_i)$ are picked up and inserted elsewhere. %
    In this case, these $j$ qubits retain their relative ordering, and thus form a reversal. This guarantees that $|\mathcal R(\tau_{i+1})| \geq j \geq \frac{1}{2} |\mathcal R(\tau_i)|$.

    \textbf{Case 2:} In the $i$th insertion step $\sigma_i$, $j < \frac{1}{2}|\mathcal R(\tau_i)|$ qubits from the set $\mathcal R(\tau_i)$ are picked up and inserted elsewhere. %
    In this case, $k = |\mathcal R(\tau_{i})| - j \geq \frac{1}{2}|\mathcal R(\tau_{i})|$ qubits are not involved in the insertion, and retain their relative order, forming a reversal. This guarantees that $|\mathcal R(\tau_{i+1})| \geq k \geq \frac{1}{2} |\mathcal R(\tau)|$.
\end{proof}

We also show the analogous proof for routing with in-order swaps:
\begin{theorem}
    For in-order swaps, the size of the largest reversal in a routing schedule can only decrease by a factor of $\frac{1}{3}$ at each step, or $|\mathcal R(\tau_{i+1})| \geq \frac{1}{3} |\mathcal R(\tau_{i})|$.
\end{theorem}

\begin{proof}
    As $\tau_{i+1}$ and $\tau_i$ are related by a single swap step, $\tau_i = \tau_{i+1} \sigma_i$, there are two possible cases.

    \textbf{Case 1:} In the $i$th swap step $\sigma_i$, $j \geq \frac{1}{3}|\mathcal R(\tau_i)|$ qubits from the set $\mathcal R(\tau_{i})$ are swapped with $j^\prime \leq j$ other qubits in $\mathcal R(\tau_i)$ and $j - j^\prime$ qubits in $S \setminus \mathcal R(\tau_i)$. %
    In this case, these $j$ qubits retain their relative ordering and are a reversal, thus $|\mathcal R(\tau_{i+1})| \geq j \geq \frac{1}{3} |\mathcal R(\tau_i)|$.

    \textbf{Case 2:} In the $i$th swap step $\sigma_i$, $j < \frac{1}{3}|\mathcal R(\tau_i)|$ qubits in $\mathcal R(\tau_i)$ are swapped %
    with $j^\prime \leq j$ other qubits in $\mathcal R(\tau_i)$ and $j - j^\prime$ qubits in $S \setminus \mathcal R(\tau_i)$. In this case, $k = |\mathcal R(\tau_i)| - j-j' \geq \frac{1}{3} |\mathcal R(\tau_i)|$
    qubits in $\mathcal R(\tau_i)$ are not involved in the permutation and remain a reversal. Therefore, $|\mathcal R(\tau_{i+1})| \geq k \geq \frac{1}{3} |\mathcal R(\tau_i)|$. 
\end{proof}
    
We have shown that the reversal monotone $|\mathcal R(\tau_i)|$ decreases by at most a factor of $\frac{1}{2}$ or $\frac{1}{3}$ in each step when implementing any permutation with riffle shuffles or in-order swaps, respectively. After the final step in implementing $\sigma$, the remaining permutation is the identity $\id$, and $|\mathcal R(\id)| = 1$. Denoting the number of steps to implement a reversal of size $k$ by $\tilde \Rt_r(k)$, we can lower bound the number of steps needed to implement any permutation $\sigma$ in 1D: $\Rt(\sigma) \geq \tilde \Rt_r(|\mathcal R(\sigma)|)$, so for some constant $c$ in each model, $\tilde \Rt_r(N) \geq 1 + \tilde \Rt_r(\frac{N}{c})$. This recurrence has the solution $\tilde \Rt_r(N) \geq \log_c (N)$. For riffle shuffles, $c=2$, so by considering a reversal of $N$ qubits, $\Rt(N) \geq \log_2(N)$. For in-order swaps, $c=3$, so $\Rt(N) \geq \log_3(N)$. We note that the constant-factor discrepancy in lower bounds results from the fact that it takes only $\log_3 N$ in-order swaps to implement the reversal. This is done by swapping the first and last thirds of the array in a single step, and then recursing on each third of the array and performing the same swap of first and last thirds until the reversal is implemented.

We can show bounds with the same asymptotic scaling in $N$
up to a weaker constant factor 
for both of these 1D models by a counting argument, which applies generally to all models of routing, as follows.

\begin{theorem}\label{thm:counting}
    For a routing model $\alpha$ over $N$ qubits with a set of $k$ single-step permutations, $\Rt^\alpha(N) \geq \frac{N \log (N / e)}{\log k} - 1$.
\end{theorem}

\begin{proof}
    For $N$ qubits, there are $N!$ possible permutations. By the pigeonhole principle, in $j$ routing steps, one may only generate at most $\frac{k^{j + 1} - 1}{k - 1} \leq k^{j  + 1}$ unique permutations. This can be seen my counting the number of nodes of a tree with branching factor $k$ and height $j$. To generate all possible $N!$ permutations, we need at least $j$ steps with 
    \begin{align}
        k^{j + 1} \geq N! &\geq \left(\frac{N}{e}\right)^N,
    \end{align}
    so
    \begin{align}
        j &\geq \frac{N \log \frac{N}{e}}{\log k} - 1
    \end{align}
    as claimed.
\end{proof}

Applying this, we see that in the case of routing with riffle shuffles %
over $N$ qubits, one starts by selecting one of $\binom{N}{i}$ size $0 \leq i \leq N$ subsets of the array, then chooses one of ${\binom{(N - i) + i}{i}}$ ways to insert the selected qubits between the remaining $N - i$ qubits. This means there are $k \leq \sum_{i=0}^{N} \binom{N}{i} \binom{N}{i} = \binom{2N}{N} \leq (2e)^{2N}$ possible single insertion steps. This shows that, for riffle shuffles, $\Rt(N) \geq \frac{\log N - 1}{2 (1 + \log 2)} - 1$. For in-order swaps over $N$ qubits, the number of single steps is upper bounded by the number of pairs of equal-sized unique subsets of the array, so $k \leq \sum_{i=0}^N \binom{N}{i}^2 = \binom{2N}{N}$. This also leads to $\Rt(N) \geq \frac{\log N - 1}{2 \left(1 + \log 2 \right)} - 1$.

\section{Results for Routing in 2D}

We now analyze two-dimensional routing.

\subsection{Protocols for 2D Routing}\label{sec:2d_protocols}

In this subsection, we present an algorithm for routing in two dimensions for both grid transfers and selective transfers. The algorithm takes advantage of a natural embedding of the $d$-dimensional hybercube graph $Q_d$ in 2D space that allows for parallel rearrangement between its hypercube subgraphs of smaller dimensions \cite{bluvstein_logical_2024}. $Q_d$ can be defined recursively by successive Cartesian graph products of the $K_2$ complete graph on two vertices: $Q_d = K_2^{\times d}$. Equivalently, one can construct the hypercube graph $Q_d$ by taking a set of vertices $V_d = \{0, \dots, 2^d - 1\}$ and connecting them via the set of edges $E_d = \{ (v_1, v_2) \mid v_1, v_2 \in V_d, H(v_1, v_2) = 1\}$ where $H(v_1, v_2)$ is the Hamming distance between $v_1$ and $v_2$ (i.e., two vertices $v_1$ and $v_2$ are connected if they differ on only one bit). %

Our embedding of the hypercube in 2D space is as follows. For the hypercube $Q_d$ on $N = 2^d$ qubits, we arrange the set of qubits on a $\sqrt{2N} \times \sqrt{N/2}$ %
grid for $d$ odd or a $\sqrt{N} \times \sqrt{N}$ grid for $d$ even. An example for $d=4$ is shown in Fig.~\ref{fig:HCEmbedding}. We label each qubit by an integer corresponding to its column or row, formed by the $\lfloor d/2 \rfloor$-bit binary representation of the row and $\lceil d/2 \rceil$-bit binary representation of the column concatenated. For example, a qubit in row $3$ and column $1$ would be written $\textcolor{red}{10}\textcolor{blue}{01}$, shown with our convention of highlighting the corresponding row bits in red and column bits in blue.
Routing utilizes the recursive structure of hypercubes. 
One can observe that the induced subgraph on $Q_d$ made by vertices that only start with $0$ forms the hypercube graph $Q_{d-1}$, and likewise for those vertices starting with $1$. We define this idea in general, as our algorithm calls for routing on these hypercube subgraphs recursively.

\begin{figure}
    \centering
    \includegraphics[width=150pt]{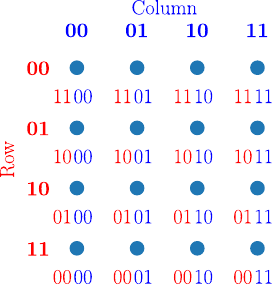}
    \caption{%
    The embedding of the hypercube $Q_4$ in a 2D array of traps. Higher-order bits (in red) correspond to the binary representation of the row and lower-order bits (in blue) correspond to the column.}
    \label{fig:HCEmbedding}
\end{figure}

\begin{definition}[Subhypercube]
    \label{def:subhypercube}
    For the hypercube $Q_d$, the subhypercube given by a bit string $a \in \{0,1\}^{*}$, where $\{0, 1\}^{*}$ indicates the set of all bit strings of any length, is the induced subgraph given by the vertices of $Q_d$ whose binary addresses start with $a$. Formally, the vertices of the subhypercube are %
    \begin{equation}
        V_a(Q_d) = \{q \in Q_d \mid \exists b \in \{0, 1\}^{d - |a|} \textrm{\ s.t.\ } q = ab \},
    \end{equation}
    where $ab$ indicates the concatenation of the binary strings $a$ and $b$. The edges of the subhypercube are those in $Q_d$ connecting pairs of vertices in $V_a(Q_d)$. One could also think of the subhypercube as the hypercube graph $Q_{d - |a|}$, with none of its edges disturbed but each vertex renamed by prepending the bit string $a$ to its binary address. Though we could also specify a hypercube subgraph of $Q_d$ by considering the vertices that all share some $i$th bit, our definition only considers fixing the leading bits as this is the useful concept for our algorithm. %
\end{definition}

Our routing algorithm schedules sets of parallel swaps between vertices of the hypercube graph that differ on some specified bit at each step. It is useful to introduce two definitions related to this concept: the \emph{cut}, which partitions the vertices into two sets depending on their $d^\prime$th bit, and the \emph{cutset}, which is the set of edges connecting vertices from different sides of the partition.

\begin{definition}[Cut and cutset of a dimension $d^\prime$ on the hypercube]
    \label{def:cut_and_cutset}
    The cut across dimension $d^\prime$ is the sets of vertices $C_{d^\prime} = (V_1, V_2)$, where $V_1$ are the vertices whose $d^\prime$th bit are 0, and $V_2$ are the vertices whose $d^\prime$th bit are 1. Formally,
    \begin{equation}
        V_1 = \{v \in Q_d \textrm{\ s.t\ } \exists a \in \{0,1\}^{d^\prime}, b \in \{0,1\}^{d - d^\prime - 1}, v = a0b \}
    \end{equation} 
    and 
    \begin{equation}
        V_2 = \{v \in Q_d \textrm{\ s.t\ } \exists a \in \{0,1\}^{d^\prime}, b \in \{0,1\}^{d - d^\prime - 1}, v = a1b \},
    \end{equation}
    where $a0b$ indicates the string made by concatenating the binary strings $a$, $0$, and $b$, and likewise for $a1b$. The cutset is the set of edges connecting pairs of vertices from $V_1$ and $V_2$ that differ only on the $d^\prime$th bit:
    \begin{equation}
        E_{d^\prime} = \{(v_1, v_2) \in E(Q_d) \mid v_1 \in V_1, v_2 \in V_2 \}.
    \end{equation}
\end{definition}

    It should be noted that the subgraph induced by vertices $V_1$ on $Q_d$ is the graph of $2^{d^\prime}$ subhypercubes given by the bit strings $a0, \forall a \in \{0,1\}^{d^\prime}$, but with additional edges connecting vertices of each subhypercube within each side of the cut. Similarly, the subgraph induced by $V_2$ on $Q_d$ is the graph of $2^{d^\prime}$ subhypercubes given by $a1, \forall a \in \{0,1\}^{d^\prime}$. In the main text, Fig.~\refmain{fig:hcembedding}(a) depicts the cut $(V_1, V_2) = C_3$ on the hypercube $Q_4$, with $V_1$ given by orange vertices and $V_2$ by blue. The edges of the cutset $E_3$ are given by dashed lines. %

Next, we show that, given any subset $E^\prime$ of the cutset of edges $E_{d^\prime}$ of $Q_d$, one can perform a swap across the pairs of qubits connected by the edges in $E^\prime$ in a single step with selective transfers, provided the qubits are arranged in our hypercube embedding. %
Each edge $e \in E^\prime$ connects pairs of vertices across the cut $C_{d^\prime} = (V_1, V_2)$. Both $V_1$ and $V_2$ correspond to the set of qubits with the $d^\prime$th bit set to $0$ or $1$, so they are formed by eliminating either a set of rows or columns from the set of all $N$ qubits, depending on the value of $d^\prime$. This means $V_1$ and $V_2$ are rectangles of the same dimension. Then, by choosing a proper masking function $\mathcal M \colon \mathbb{Z} \to \{0, 1\}$, the selective transfers swap step given by $V_1$, $V_2$, and $M$ will swap all of the atoms across the desired set of edges $E^\prime$.
To perform the same set of swaps across $E^\prime$ using only grid transfers, we instead must break down the subsets of atoms of $V_1$ and $V_2$ which are being swapped into sets of disjoint rectangles, and perform a series of swaps with these. The problem of decomposing sets of qubits into rectangles is studied in Ref.~\cite{tan_depthoptimal_2024}, where an optimal algorithm utilizing a theorem prover and polynomial-time heuristic are presented.
Additionally, it is known that, for a rectangular array of points of size $m \times m^\prime$, there exist subsets needing at least $\min(m, m^\prime)$ rectangles to be described. %
As we embed in a $\sqrt{2N} \times \sqrt{N/2}$ or $\sqrt{N} \times \sqrt{N}$ grid, breaking down some selective transfers swap steps into swaps of disjoint rectangles generally will require $\Omega(\sqrt N)$ rectangles, and thus the naive method of decomposing a point array into a single rectangle per row or a single rectangle per column never performs worse than the worst case of an optimal algorithm. %
For this reason, we will see that there is often a factor of $\sqrt{N}$ advantage in routing with selective transfers as compared to grid transfers. This advantage does not hold, however, in the case of sparse routing, where, as we show in Section~\ref{sec:sparse}, there is only a polylogarithmic separation between the models' routing numbers for this task. %

\begin{figure*}
\begin{algorithm}[H]
    \caption{Hypercube routing with selective transfers}
    \label{alg:hypercuberouting}
j    \begin{algorithmic}
        \LComment{Subroutine to fix errors between%
        subhypercubes given by bit-strings $a0$ and $a1$. Described in Sec.~\ref{sec:matching_appendix}}
        \Procedure{Cut-Errors}{$a \in \{0, 1\}^*$, $d \in \mathbb{Z}$, $\sigma\colon [2^d] \to [2^d]$}
            \State $V_{c1} \gets$ the vertices of the subhypercube given by $a0$
            \State $V_{c2} \gets$ the vertices of the subhypercube given by $a1$
            \LComment{$E_c$ is the subset of the cutset across dimension $|a|$ connecting subhypercubes $a0$ and $a1$}
            \State $E_c \gets \{(v_1, v_2) \in E(Q_d) \mid v_1 \in V_{c1} \land v_2 \in V_{c2}\}$
            \State $G = (V,E) \gets (V_{c1} \cup V_{c2}, \emptyset)$
            \State Mark the edges in $E_c$ green and add them to $E$
            \State $C \gets \{(i,j) \mid \exists b \in \{0,1\}^{d - 1 - |a|}, \sigma(i) = a0b \land \sigma(j) = a1b\}$
            \ForAll{$(i,j) \in C$}
                \If{$i$ and $j$ are on the same side of the cut given by $a$}
                    \State Add $(i,j)$ to $E$, marked red
                \Else
                    \State Add $(i,j)$ to $E$, marked blue
                \EndIf
            \EndFor

            \LComment{Now we construct the error graph, and perform perfect matching}
            \State $G_E = (V_E, E_E)$
            \State $V_E \gets \{e \in E \mid \textrm{e is marked red}\}$
            \ForAll{Paths $P \subset E$ on $G$ of green and blue edges connecting pairs of red edges $e_1,e_2 \in E$}
                \State $E_E \gets (e_1, e_2)$
            \EndFor

            \State $E_p \gets $ the edges of a perfect matching on the graph $G_E$
            \State $\sigma_E \gets \id$
            \ForAll{$e = (v_1, v_2) \in E_p$}
                \LComment{Note: $v_1, v_2$ are red edges in $E$}
                \State $P \gets $ the green edges of a path of only green and blue edges connecting the edges $v_1, v_2 \in E$
                \State $\sigma^\prime \gets $ the permutation corresponding to swaps along the edges of P
                \State Perform $\sigma^\prime$ \Comment{This corresponds to a single selective transfers step}
                \State $\sigma_E \gets \sigma^\prime \sigma_E$ \Comment{We don't have to worry about commutativity issues, as all the P selected in this loop will be disjoint from each other}
            \EndFor
            \State \Return $\sigma_E$
        \EndProcedure

        \LComment{Generates a routing procedure for the permutation $\sigma$ using selective transfers. Can be readily converted to one using grid transfers.}
        \Procedure{Hypercube-Routing}{$a \in \{0,1\}^*, d \in \mathbb{Z}, \sigma\colon [2^d] \to [2^d]$}
            \State $\sigma_1 \gets $\Call{Cut-Errors}{$a$, $d$, $\sigma$}%
            \State $\sigma^\prime \gets \sigma_1^{-1} \sigma$ \Comment{The remaining permutation to route $\sigma$}            \LComment{Note: The next two calls can be parallelized with selective transfers}
            \State $\sigma_2 \gets $ \Call{Hypercube-Routing}{$a0$, $d$, $\sigma^\prime$}
            \State $\sigma_3 \gets $\Call{Hypercube-Routing}{$a1$, $d$, $\sigma^\prime$}

            \State $\sigma^{\prime} \gets \sigma_3^{-1} \sigma_2^{-1} \sigma_1^{-1} \sigma$ \Comment{The remaining permutation}
            \State $E_c \gets $ the cutset on $Q_d$ given by $a$
            \LComment{At this point, the only errors remaining in the hypercube are pairs of atoms across the cutset whose high order bits are on the wrong cut side}
            \State $E_E \gets \{(v_1, v_2) \in E_c \mid v_1 \neq \sigma^\prime(v_1) \}$
            \State $\sigma_f \gets $ the permutation corresponding to swaps along the set of edges $E_E$ 
            \State Perform $\sigma_f$ \Comment{This is one selective transfers step}
            \State \Return $\sigma_1 \sigma_2 \sigma_3 \sigma_f$
        \EndProcedure
    \end{algorithmic}
\end{algorithm}
\end{figure*}

Finally, with these primitives, we show an algorithm (Alg.~\ref{alg:hypercuberouting}) for routing with selective transfers based on prior works in hypercube graph routing. Its output can be converted to a routing schedule for routing with grid transfers. %
It is shown in Ref.~\cite{alon_routing_1993} that, in the case of routing by swaps on a coupling graph, $\Rt (G \times G^\prime) \leq 2 \Rt(G^\prime) + \Rt(G)$. Routing in such a product graph can be performed as follows. First, route in parallel across the subgraphs corresponding to copies of $G$ so that each subgraph corresponding to a copy of $G^\prime$ contains a set of vertices representing all addresses from $G^\prime$ (the existence of a set of permutations that can be performed across the copies of $G$ to satisfy this requirement is guaranteed by Hall's marriage theorem). Then, route on the copies of $G^\prime$ to put the vertices' $G^\prime$ addresses in order. Finally, route in the copies of $G$ to correct all the vertices' $G$ addresses. This completes routing of $G \times G^\prime$.

We show this algorithm can be scheduled to run in $2d - 1$ steps for our embedding of $Q_d = K_2 \times Q_{d-1}$. The first step of the algorithm performs a set of swaps across the copies of $K_2$, which is a selected subset of the cutset edges $E_0$ of the graph. We have already shown this to be a single selective transfers routing step. This subset is chosen by a matching procedure so that each side of the cut contains a set of vertices whose last $d - 1$ bits form the set of all $d-1$ length bit strings. We outline this matching procedure in Sec.~\ref{sec:matching_appendix}. Following this, the algorithm recurses, calling for routing on each $Q_{d-1}$ subgraph in the cut $C_0$. This sorts each subhypercube on each side of the cut so that the last $d - 1$ bits of each atom's destination match its location address, but leaves the first bit of each address potentially out of place. We perform one more swap across a subset of edges of the cutset $E_0$, swapping those atoms whose first bits of their destination don't match their location to put them in place. Again this is a single selective transfers step. 

Note that besides the recursive step, the routing algorithm only ever calls for swaps across a subset of the cutset $E_0$. Thus, when the recursive step is called, the algorithm considers the subhypercubes given by bit strings $0$ and $1$, and performs swaps across subsets of $E_1$, $E_2$, $E_3$, etc. By scheduling these swaps so all of the swaps across $E_1$ occur at the same step (and likewise for $E_2$, $E_3$, etc.), routing on each subhypercube across the cut is done in parallel, and each set of swaps across each cutset $E_i$ is done in parallel. Since the procedure makes one parallel recursive call and two selective transfers steps, it routes in $T(d) = 2 + T(d - 1)$ selective transfers steps, which means $T(d) = 2d - 1$ (equivalently, $T(N) = 2 \log_2 N - 1$). This shows that $\Rt(N) \leq 2 \log_2 N - 1$ with selective transfers. To route with only grid transfers, one converts each selective transfers step from this algorithm into a series of grid transfers as outlined earlier, leading to a $\sqrt N$ increase in the number of steps. Thus, for grid transfers, $\Rt(N) \leq \sqrt N (2\log_2 N - 1)$. 

\subsection{Matching Procedure}\label{sec:matching_appendix}

In this subsection, we describe the matching procedure mentioned in Alg.~\ref{alg:hypercuberouting}. An example of a routing sequence on the hypercube $Q_4$, including the matching step, is shown in Fig.~\ref{fig:RoutingDesc}. In Fig.~\ref{fig:RoutingDesc}(a), the atoms are depicted in the 2D embedding, with their destination address indicated by the label the arrows point to. We underline the first bit to indicate that this is the highest level of the recursive routing step, and that we will depict the steps that perform swaps across the cutset $E_1$ (the set of edges connecting atoms that differ on the first bit). The remaining bits are overlined to indicate that they represent the destination address in the subhypercube they are being routed to. The matching procedure begins in Fig.~\ref{fig:RoutingDesc}(b). Green edges are added to connect vertices of the hypercube that differ on the first bit (note, we refer to the vertex address, not its destination address marked by the arrow). Then, blue edges are added to connect atoms whose destination addresses differ on the first bit. This creates a series of loop graphs. We highlight one loop graph by leaving its edges solid and making the other loops dashed.

In the next step, shown in Fig.~\ref{fig:RoutingDesc}(c), we continue matching. If two qubits are connected by a blue edge, but reside on the same side of the cut $C_1$, we consider them to be an ``error'' and convert this blue edge to a red edge. Then, on each loop graph, a perfect matching is performed, and green edges found to be in the perfect matching are marked orange. There are only two possible perfect matchings on each loop, as one must choose the green edges of alternating paths connecting red edges, since, in a perfect matching, each red error edge only has one matched edge connecting it. For any loop, the overlined part of the qubits' destinations on each side of the cut form a family of two sets (in the highlighted loop in Fig.~\ref{fig:RoutingDesc}(b), %
$\{\{100, 100\}\}, \{000, 000\}\}$). This set always satisfies Hall's marriage condition, so a perfect matching corresponding to a traversal of the set always exists.

Figure~\ref{fig:RoutingDesc}(d) shows the result of swapping across the matching edges found in the last step. The sets of overlined addresses on each side of the cut $C_1$ across the first bit in the array are now unique. In Fig.~\ref{fig:RoutingDesc}(e), routing on the subhypercubes given by the addresses $0$ and $1$ is performed, thus sorting the atoms.  %
Pairs of qubits that are on the wrong side of the cut $C_1$ are connected by an edge again. Figure~\ref{fig:RoutingDesc}(f) shows the array after these edges are swapped, thus sending all atoms to their destination.

\begin{figure}
    \centering
    \includegraphics[width=0.8\linewidth]{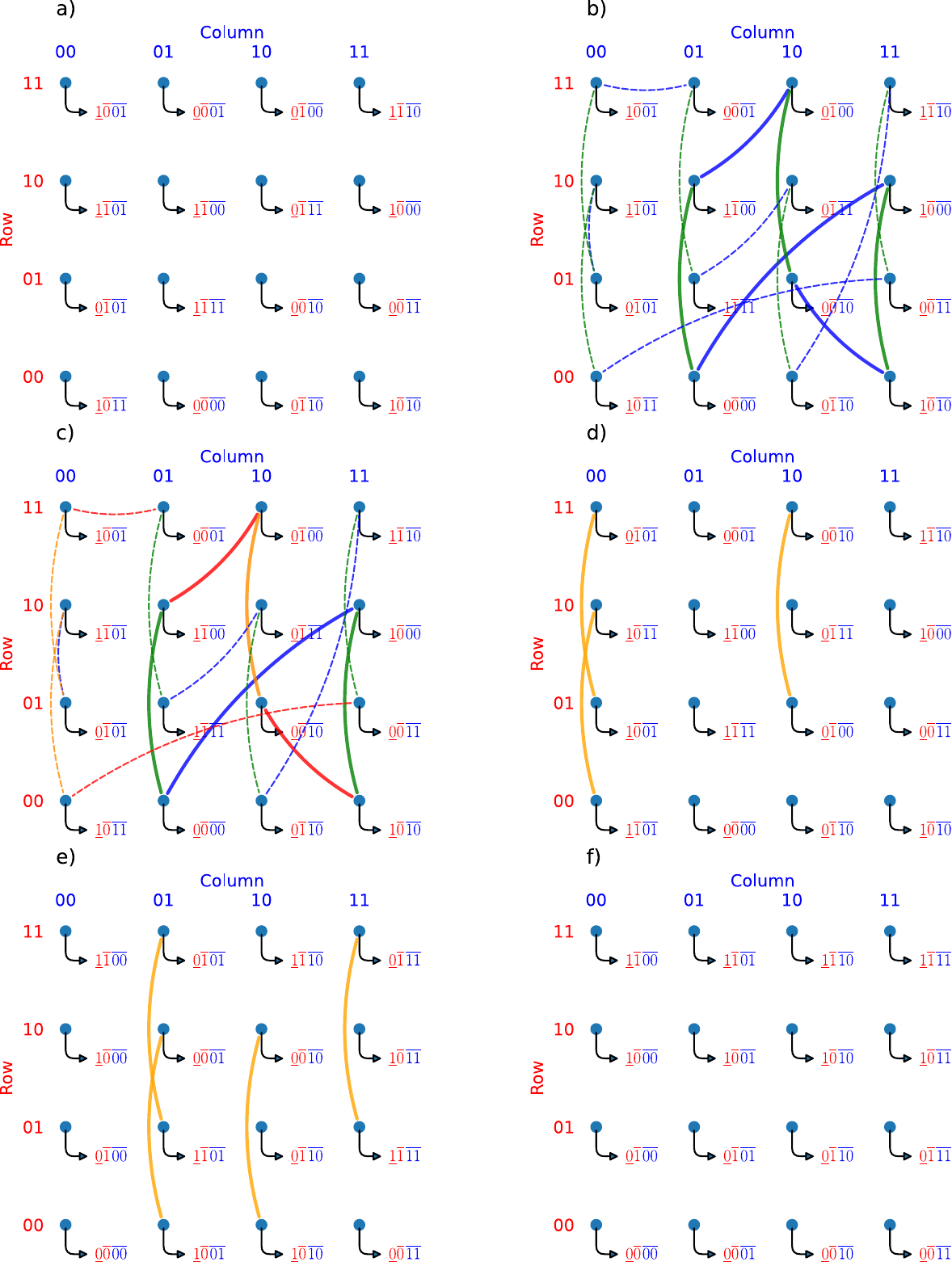}
    \caption{Illustration of the matching procedure for hypercube routing. 
    \textbf{(a)} The vertices of the atoms shown in the hypercube embedding, with arrows pointing to the binary vertex addresses of their destination. Underlined bits indicate the side of the cut $C_0$ (Def.~\ref{def:cut_and_cutset}) of the destination. Overlined bits indicate the address of each destination in the subhypercube (Def.~\ref{def:subhypercube}) it is routed to. \textbf{(b)} Loops are formed by adding green edges between each vertex on each side of the cut, and blue edges between vertices whose destinations have the same (overlined) subhypercube address. \textbf{(c)} Blue edges that connect vertices on the same side of the cut $C_0$ are considered errors, and marked red. Green edges that are part of the perfect matching are marked orange. \textbf{(d)} All other edges are removed to highlight only the edges of the perfect matching. Swaps between qubits connected by matched edges are performed. \textbf{(e)} Routing is called recursively on each subhypercube of the cut $C_0$. Only first (overlined) bits remain out of order, so swaps (marked by orange edges) are scheduled to swap qubits into their destination. \textbf{(f)} Routing is completed.
    \label{fig:RoutingDesc}}
\end{figure}

\subsection{Lower Bounds for 2D routing}\label{sec:2d_bounds}
In this subsection, for both models of 2D routing, we derive lower bounds on the routing number by applying Theorem~\ref{thm:counting}. 
\begin{theorem}\label{thm:counting_2DRect}
    For 2D routing with grid transfers, $\Rt(N) = \Omega(\sqrt N \log N)$ for a square grid of $N = m \times m$ atoms.
\end{theorem}
\begin{proof}
    When routing with grid transfers, each single step is a swap between two disjoint rectangles. The number of pairs of rectangles in a $m \times m$ grid upper bounds the number of swap steps, so $k \leq \sum_{i,j} \left( \binom{m}{i} \binom{m}{j} \right)^2 =  \binom{2m}{m} \leq (2e)^{4m}$, as each rectangle is specified by a subset of $m$ rows and columns. Applying Theorem~\ref{thm:counting}, we find $\Rt(N) \geq \frac{\sqrt N (\log N - 1)}{4 (1 + \log 2)} - 1$.
\end{proof}

\begin{theorem}
    For 2D routing with selective transfers, $\Rt(N) = \Omega(\log N)$.
    \label{thm:counting_2dselective}
\end{theorem}
\begin{proof}
    Each single step with selective transfers is specified by two combinatorial rectangles $R_1$ and $R_2$ and a masking function $\mathcal M \colon \mathbb{Z} \to \{0, 1\}$. Since the step results in a swap of two subsets of the array, the number of single steps $k$ is upper bounded by the number of equal sized subsets of the array, just as in 1D routing with in-order swaps. Thus, $k \leq (2e)^{2N}$ and
    $\Rt{(N)} \geq \frac{\log N - 1}{2 (1 + \log 2)} - 1$.
\end{proof}

In two dimensions, the swap operation with selective transfers retains the basic fact that there are two sets of atoms that are picked up and remain in order with respect to each other in each set. This suggests that our lower bound in one dimension applies to two dimensions as well, as we now demonstrate.

Consider implementing a reversal on a subset of a single row of atoms of size $n$ in an infinite two-dimensional grid of empty static traps. Since we are in 2D, the set of qubits is some set $S^{2D} \subseteq \mathbb{Z} \times \mathbb{Z}$. Once again we may define a reversal monotone on this array of the largest subset of atoms that are in unique columns, and are being routed to in reverse order to another set of unique columns. We still let $\sigma$ denote the target permutation, implemented by steps $\sigma = \sigma_k \sigma_{k-1} \dots \sigma_1$, and denote the remaining permutation at each step by $\tau_i = \prod_{j=k}^{i} \sigma_j$.
\begin{definition}[2D Reversal Monotone]
    For qubits in a 2D array with coordinates given by $S^{2D} \subseteq \mathbb{Z} \times \mathbb{Z}$, and a permutation $\sigma\colon S^{2D} \to S^{2D}$, the 2D reversal monotone $\mathcal R^{2D}(\sigma)$ is the largest subset $x \subseteq S^{2D}$ that satisfies the following:
    \begin{equation}
        \forall i,j \in x, i \neq j \implies \col(i) \neq \col(j) \textrm{\ and\ } \col(i) < \col(j) \implies \col(\sigma(i)) > \col(\sigma(j))
    \end{equation}

\end{definition}

Here $\col$ indicates the column of a coordinate tuple, i.e., $\col((i,j)) = j$. We can now see clearly that the same proof steps restrict this monotone to decrease only by a factor of $\frac{1}{3}$ at each step of any routing sequence.

\begin{theorem}
    $|\mathcal R^{2D}(\tau_{i+1})| \geq \frac{1}{3} |\mathcal R^{2D}(\tau_i)|$ for 2D routing with selective transfers.
\end{theorem}

\begin{proof}

\textbf{Case 1:} In a single step, $j \geq 2|\mathcal R^{2D}(\sigma)|/3$ qubits in $\mathcal R^{2D}(\sigma)$ are swapped. This must be done by masking the swap of two combinatorial rectangles, such that two sets of atoms $A$ and $B$ are swapped. Both $A$ and $B$ have $j/2$ qubits, and the relative ordering of the qubit's columns within both $A$ and $B$ must remain the same, so these $j/2$ qubits still form a 2D reversal. Therefore, $|R^{2D}(\tau^{i+1})| \geq j/2 \geq \frac{1}{3} |R^{2D}(\tau_i)|$.

\textbf{Case 2: } Otherwise, $j < 2 |\mathcal R^{2D}(\sigma)|/3$ qubits in $\mathcal R^{2D}(\sigma)$ are involved in the swap. These $j$ qubits form a reversal, so $|\mathcal R^{2D}(\tau^{i+1})| \geq j \geq \frac{1}{3} |\mathcal R^{2D}(\tau_i)|$.
\end{proof}

Therefore, when routing with selective transfers, implementing the single reversal of a row of $m$ atoms has the recurrence $T^{2D}(m) \geq 1 + T^{2D}(m/3)$, regardless of how many rows are in the grid. Therefore $\Rt(N) \geq \log_3 (\sqrt N)$. Asymptotically, this lower bound improves upon the bound in Theorem~\ref{thm:counting_2dselective} by a factor of $\frac{1 + \log 2}{\log 3} \approx 1.54$. %

\subsection{Sparse Routing with grid transfers}\label{sec:sparse}

In this subsection, we discuss sparse routing. While the task of general routing with grid transfers takes $\Omega(\sqrt N \log N)$ steps for most permutations, this model turns out to be well suited to the task of sparse routing, where only a small number of qubits are permuted. We formalize this in terms of the number of qubits per row and column that will be routed.

\begin{definition}[Column-sparse and row-sparse permutations]\label{def:sparse}
    A column-sparse or row-sparse permutation $\sigma\colon S \to S$ over $S = \{1,2, \dots, N\}$ is a permutation that is nonidentity for $\bigo{\poly(\log N)}$ qubits per column or row, respectively. %
\end{definition}

We show that any column-sparse or row-sparse permutation can be implemented with $O(\poly(\log N))$ grid transfer steps via the following procedure.

\begin{theorem}\label{thm:sparseRting}
    For a column-sparse or row-sparse permutation $\sigma$, $\Rt(N, \sigma) = \bigo{\poly(\log N)}$ with grid transfers. %
\end{theorem}

\begin{proof}
    In the targeted sparse permutation $\sigma$ on a grid of size $m \times m$, $N = m^2$, consider the set of qubits $A \subset S^{2D} = \mathbb{Z}^m \times \mathbb{Z}^m$ for which $\forall i \in A, \sigma(i) \neq i$. We assume $\sigma$ to be column-sparse in the sense that $A$ has $\bigo{\poly(\log N)}$ qubits per column. We describe a procedure to implement $\sigma$ in Algorithm \ref{alg:sparse}. The general outline of the procedure is to choose a set of qubits $i, \sigma(i) \in S^{2D}$ that come from a set of independent columns such that no more than a qubit $i$ and its destination $\sigma(i)$ are in the same column in the set. From there, a compression procedure packs these qubits into two adjacent rows, where 1D routing is performed, and then by reversing the compression procedure half of the selected qubits will have been routed to their destination. This process is repeated until $\sigma$ is realized.

    \begin{figure*}
    \begin{algorithm}[H]
        \caption{Column-sparse routing}
        \label{alg:sparse}
        \begin{algorithmic}
            
            \Require $\forall (i,j), (i^\prime, j^\prime) \in A, j = j^\prime \implies i = i^\prime$
            \Comment{Ensure qubits are in unique columns}
            \Require $\forall b \in B, \col{(b)} \geq k$ \Comment{Ensure qubits start in columns $\geq k$}
            \LComment{Compress the qubits in $A$ to row $k$}
            \Procedure{Compress}{$B \subseteq S^{2D}$, $m \in \mathbb{Z}$, $k \in \mathbb{Z}$}
                \If{m = 0}
                    \State \Return $\id$
                \EndIf
                \LComment{The set of columns with a qubit in B above the row midpoint of the array}
                \State $C \gets \{j \mid (i,j) \in A, i > \lfloor\frac{m}{2}\rfloor + k\}$ 
                \LComment{The two sets of rows, either above or below the midpoint}
                \State $R_1 \gets \{ k, \dots, \lfloor m/2 \rfloor + k \}$
                \State $R_2 \gets \{ \lfloor m/2 \rfloor + k + 1, \dots, 2 \lfloor m/2 \rfloor + k\}$
                \State $\sigma_1 \gets $ the permutation corresponding to a swap of rectangles $R_1 \times C$ and $R_2 \times C$
                \State Perform $\sigma_1$
                \State $B^\prime \gets \{\left(i, \left(j - k \bmod \lfloor \frac{m}{2} \rfloor\right) + k) \mid (i,j) \in B\right)\}$ \Comment{The locations of the qubits in B after the swap.}
                \State $\sigma_2 \gets$ \Call{Compress}{$B^\prime$, $\lfloor \frac{m}{2} \rfloor$, $k$}
                \State \Return $\sigma_2 \sigma_1$
            \EndProcedure

            \Require $\sigma$ is a column-sparse permutation as in definition \ref{def:sparse}
            \Procedure{SparseRoute}{$\sigma\colon S^{2D} \to S^{2D}$}
                \If{$\sigma = \id$}
                    \State \Return id
                \EndIf
                \State $A \gets \{ q \in S^{2D} \mathrm{\ s.t\ } q \neq \sigma(q)\}$
                \Comment{The set of qubits addressed non-trivially by $\sigma$}
                \State Mark all of the qubits in $A$ black
                \State $B_1 \gets \emptyset$
                \State $B_2 \gets \emptyset$
                \While{$\exists q \in A$ s.t. $q$ and $\sigma(q)$ are black}
                    \State $B_1 \gets B_1 \cup \{q\}$
                    \State $B_2 \gets B_2 \cup \{\sigma(q)\}$
                    \State $E \gets \{ q^\prime \in A \mathrm{\ s.t\ } \{\col(q^\prime), \col(\sigma(q^\prime))\} \cap \{\col(q), \col(\sigma(q))\} \neq \emptyset \} $
                    \State Mark the qubits in $E$ red
                \EndWhile
                \State $\sigma_1 \gets$ \Call{Compress}{$B_1$, $m$, $1$}
                \State $B_2^\prime \gets \{\sigma_1 b \mid b \in B_2\}$ \Comment{Keep track of how the qubits in $B_2$ move during the compression of $B_1$}
                \State $\sigma_2 \gets$ \Call{Compress}{$B_2^\prime$, $m - 1$, $2$} 
                \State Perform 1D routing on the first row so that all the original pairs of $q \in B_1, \sigma(q) \in B_2$ are in the same column, and store this permutation as $\sigma_3$
                \State Swap the first and second row, and store this permutation as $\sigma_4$
                \State Perform the inverse routing schedules $\sigma_2^{-1} \sigma_1^{-1}$
                \LComment{At this point, the qubits $B_1$ will have been routed to their destination}
                \State $\sigma^\prime \gets \sigma_1 \sigma_2 \sigma_4^{-1} \sigma_3^{-1} \sigma_2^{-1} \sigma_1^{-1} \sigma$ \Comment{The remaining permutation to be performed}
                \State \Call{SparseRoute}{$\sigma^\prime$}
            \EndProcedure
        \end{algorithmic}
    \end{algorithm}
    \end{figure*}

    The sparse routing protocol first marks all of the qubits in $A$ black, indicating that they have not been eliminated yet. It then iteratively chooses any pair of black qubits $q \in A$ and $\sigma(q)$, and adds them to sets $B_1$ and $B_2$, respectively. Next, it eliminates any qubits that might conflict with the routing of $q$ to $\sigma(q)$, marking red any qubit $q^\prime$ and its destination $\sigma(q^\prime)$ where either $q^\prime$ or $\sigma(q^\prime)$ share a column with $q$ or $\sigma(q)$. Once the protocol is no longer able to select any more pairs of black qubits $q,\sigma(q)$, it compresses the qubits in $B_1$ into a single row, as well as those in $B_2$ to an adjacent row. After this, 1D routing is performed on the row corresponding to $B_1$, bringing together pairs $q,\sigma(q)$ in the same column. Finally, the two rows are swapped, and compression is reversed, thus swapping qubits $q \in B_1$ with their destination $\sigma(q) \in B_2$. The algorithm recurses on the remaining permutation to be performed, which is still guaranteed to meet the sparseness requirements. Though the compression procedure scrambles the qubits outside of $B_1$ and $B_2$ within their columns, as we reverse it at the end of the round no qubits other than $B_1$ and $B_2$ end up moving.

    To see that this algorithm implements $\sigma$ in $\bigo{\poly(\log N)}$ steps, we show that the compression procedure generates $\bigo{\log m}$ steps at most, and that the routing procedure selects sufficiently large sets $B_1$ and $B_2$ %
    per round such that only $\bigo{\poly(\log N)}$ rounds must be performed. For each compression procedure, compression of a set of qubits in independent columns is performed by selecting those columns containing qubits in rows $> \lfloor m/2 \rfloor$, and swapping the top and bottom halves of the column. This brings all of the qubits into rows $\leq \lfloor m/2 \rfloor$, and then this procedure recurses, folding selected columns in quarters, etc., until all of the qubits are in the first row. This clearly takes at most $\log_2 m$ steps.

    To justify the size of the routing procedure's selection of qubits, observe that each time a pair of qubits $q, \sigma(q)$ are added to $B_1$ and $B_2$, at most $\bigo{\poly(\log m)}$ qubits are marked red. This is because each column can only contain $\poly(\log m)$ elements affected by the permutation by definition, so when the two columns $\col(q)$ and $\col(\sigma(q))$ are eliminated, there can only be $\poly(\log m)$ pairs of qubits being routed into or out of these columns. As there are  $k = O(m \poly(\log m)$ elements in the permutation, the procedure routes $k / \bigo{\poly(\log m)}$ elements into place per round, thus requiring $\frac{k}{k / \bigo{\poly(\log m)}} = \bigo{\poly(\log m)}$ rounds total.

    Overall, each round uses at most $2 \log_2 m$ steps to perform two compressions and $\log_2 m$ steps of 1D routing, a single-step swap, and then at most $2\log_2 m$ more steps to reverse the two compressions. This makes the total steps generated $\bigo{\poly(\log m)} \left( 5\log_2 + 1 \right) = \bigo{\poly(\log m)}$. Thus, for any sparse permutation $\sigma$, $\Rt_\sigma = \bigo{\poly(\log m)} = \bigo{\poly(\log N)}$.
\end{proof}

\bibliographystyle{unsrtnat}
\bibliography{references.bib}
\makeatletter\@input{xx_main.tex}\makeatother